\documentclass[journal]{IEEEtran}
\usepackage{lineno}
\usepackage{float}
\usepackage{cite}
\usepackage{hyperref}
\usepackage{amsmath,amssymb,amsfonts}
\usepackage{amsthm}
\DeclareMathOperator*{\argmax}{argmax}
\usepackage{graphicx}
\usepackage{textcomp}
\usepackage{xcolor}
\usepackage{graphicx}
\usepackage{subfigure}
\usepackage{amsmath}
\usepackage{amsfonts,amssymb}
\usepackage{mathrsfs}
\usepackage{mathtools}
\usepackage{algorithm}
\usepackage{algorithmicx}
\usepackage{algpseudocode}
\usepackage{bm}
\usepackage{multirow}
\usepackage{array}
\usepackage{url}
\usepackage{fancyhdr}
\usepackage{mdwmath}
\usepackage{mdwtab}
\usepackage{caption}
\usepackage{setspace}
\usepackage{bm}
\usepackage{algpseudocode}
\usepackage{mathtools}
\usepackage{dsfont}
\usepackage{bbm}
\usepackage{cases}
\newtheorem{remark}{Remark}
\theoremstyle{definition}
\newtheorem{theorem}{Theorem}

\newtheorem{lemma}{Lemma}

\newtheorem{corollary}{Corollary}

\makeatletter
\newcommand{\biggg}{\bBigg@{3}}
\newcommand{\Biggg}{\bBigg@{3.5}}
\makeatother
\expandafter\def\expandafter\normalsize\expandafter{%
    \normalsize%
    \setlength\abovedisplayskip{4pt}%
    \setlength\belowdisplayskip{4pt}%
    \setlength\abovedisplayshortskip{2pt}%
    \setlength\belowdisplayshortskip{2pt}%
}
\begin{document}

\title{Continuous-Aperture Array-Based ISAC \\Over Fading Channels}

\author{Boqun Zhao, \IEEEmembership{Graduate Student Member, IEEE}, Chongjun Ouyang, \IEEEmembership{Member, IEEE},\\Xingqi Zhang, \IEEEmembership{Senior Member, IEEE}, and Yuanwei Liu, \IEEEmembership{Fellow, IEEE}\vspace{-10pt}
	\thanks{B. Zhao and X. Zhang are with the Department of Electrical and Computer Engineering, University of Alberta, Edmonton AB, T6G 2R3, Canada (email: \{boqun1, xingqi.zhang\}@ualberta.ca).}
	\thanks{C. Ouyang is with the School of Electronic Engineering and Computer Science, Queen Mary University of London, London, E1 4NS, U.K. (e-mail: c.ouyang@qmul.ac.uk).}
	\thanks{Y. Liu is with the Department of Electrical and Electronic Engineering, The University of Hong Kong, Hong Kong, and also with the Department of Electronic Engineering, Kyung Hee University, Yongin-si, Gyeonggi-do 17104, South Korea (e-mail: yuanwei@hku.hk).}
}

\maketitle

\begin{abstract}
A framework of continuous-aperture array (CAPA)-based integrated sensing and communications (ISAC) under a fading communication channel is proposed. A continuous operator-based signal model is developed, and the statistics of the communication channel gain are characterized via Landau's eigenvalue theorem. On this basis, the performance of the CAPA-based ISAC system is analyzed by considering three continuous beamforming designs: i) the sensing-centric (S-C) design that optimizes sensing performance, ii) the communication-centric (C-C) design that optimizes communication performance, and iii) the Pareto-optimal design that balances the sensing-communication trade-off. For the S-C and C-C design, closed-form expressions for the sensing rate (SR), ergodic communication rate (CR), and outage probability are derived, and high-signal-to-noise ratio asymptotic analysis is conducted to obtain the multiplexing and diversity gains. For the Pareto-optimal design, the Pareto-optimal beamformer achieving the Pareto boundary is derived, and the achievable SR-CR region is characterized. Numerical results demonstrate that the proposed CAPA-ISAC scheme outperforms both conventional spatially discrete arrays-based ISAC and CAPA-based frequency-division sensing and communications.
\end{abstract} 
\section{Introduction}

Integrated sensing and communications (ISAC) has emerged as a key enabler for future wireless networks, driven by the need to reuse spectrum and hardware while supporting perception-centric applications such as localization, tracking, and environment mapping \cite{liu2022integrated}. By allowing sensing and communications to share waveform, spectrum, and infrastructure, ISAC promises improved resource efficiency and new system-level synergies, but also introduces intrinsic conflicts between the two functionalities in waveform/beamforming design and performance evaluation. A growing body of work has therefore focused on establishing performance metrics, fundamental trade-offs, and design principles for dual-functional systems, ranging from information-theoretic formulations to signal processing–oriented transceiver optimization \cite{LiuAn,liu2022integrated}. In particular, mutual-information (MI)-based metrics provide a unified way to quantify both sensing and communication capabilities and to characterize their trade-off under a common information measure \cite{ouyang2023integrated,bcrb}.

{Among the key research directions in ISAC, multiple-antenna technologies have received significant attention because they provide spatial degrees of freedom (DoFs) and array gains for both communication and sensing, thereby improving throughput and sensing resolution \cite{liu2022integrated}. Recent studies have further explored flexible array architectures \cite{shao2025sixd,li2026reconfig}, which exploit spatial reconfigurability to enhance ISAC performance. Despite their different implementations, these architectures follow a common evolutionary trend toward larger apertures, denser antenna deployments, higher operating frequencies, and more flexible structures. This trend leads naturally to an (approximately) continuous electromagnetic (EM) aperture, referred to as a continuous-aperture array (CAPA), whose radiating structure is modeled by a continuous current distribution over a spatial aperture \cite{liu2024capa}.}

{While the CAPA relies on an idealized continuous-current model, this abstraction is grounded in realizable hardware. Functionally, a CAPA is an analog EM-beamforming architecture that synthesizes a continuous source current at the radio-frequency front end. It also represents the continuous limit of a densely-packed array as the inter-element spacing vanishes. The pursuit of such continuous apertures traces back to the infinite-current-sheet concept of the 1960s \cite{wheeler1965}, while modern implementations rely on two key capabilities. First, per-point control of the radiated amplitude and phase has been enabled by reconfigurable metasurfaces, including holographic \cite{huang2020holographic}, dynamic \cite{shlezinger2021dynamic}, and waveguide-fed \cite{smith2017analysis} designs. Second, near-spatial continuity has been approached by leaky-wave \cite{araghi2021holographic}, interdigital-transducer \cite{yuan2024interdigital}, and tightly-coupled \cite{prather2017optically} structures that densely approximate a continuous current sheet. These advances have led to functional and even commercially deployed prototypes \cite{sazegar2022full,pivotal2019holographic}, while emerging tri-hybrid architectures combine digital, analog, and EM beamforming to approach an ideal CAPA \cite{heath2025trihybrid}; see \cite{liu2024capa} for a comprehensive overview. The idealized CAPA therefore characterizes the fundamental performance limit, i.e., an upper bound, for these dense and (semi-)continuous apertures, and the optimized current distributions can be spatially sampled to guide the design of practical ultra-dense arrays. }

The continuous viewpoint introduced by CAPA reshapes both the modeling and the optimization: instead of finite-dimensional beamforming vectors used in spatially discrete arrays (SPDAs), CAPA transmit/receive processing is more naturally described as functional beamforming defined over a continuous spatial domain. Consequently, the fundamental spatial dimensionality is governed by signal-space arguments and the spectral properties of aperture-domain operators, which connect CAPA performance limits to spatial DoFs characterized by operator spectra and spatial bandwidth notions \cite{poon2005degrees,pizzo2022nyquist,dardari}. This operator-based perspective also opens the door to analyzing CAPA systems under physically grounded stochastic channel models, where the channel can be represented as an EM random field \cite{pizzo2022spatial}. These features make CAPA a natural candidate for ISAC, where fine spatial focusing and controllable beampatterns are highly valuable for both communication coverage and sensing resolution.

ISAC has been extensively studied in terms of waveform/precoder design, performance trade-offs, and system architectures, as summarized in major surveys and tutorials \cite{liu2022integrated,LiuAn}. In particular, MI-based sensing has been used to establish information-theoretic limits and to connect sensing quality directly to waveform/beamformer design under a unified information measure \cite{ouyang2023integrated,bcrb}. These studies provide important foundations for joint design, but many of them adopt finite-dimensional discrete-array models, which cannot directly carry over to continuous-aperture architectures.

Recent studies on CAPAs establish EM-consistent channel models and continuous-domain transmission frameworks in which beamforming becomes a functional design and performance is governed by the spectra of spatial correlation operators. In particular, fundamental-limit work on continuous-aperture surfaces provides useful physical insights and DoF-type interpretations that motivate continuous-domain formulations \cite{dardari,per3}. Further, EM-consistent random channel characterizations provide a principled statistical basis for continuous-aperture links and naturally lead to operator-based correlation descriptions \cite{pizzo2022spatial}. Building on these foundations, capacities for CAPA communications have been characterized using operator-based formulations \cite{zhao2024continuous}, and multipath-induced spatial fading models and associated performance metrics (e.g., diversity and multiplexing gains) have also been investigated \cite{ouyang_diversity}. Moreover, a growing body of work studies CAPA beamforming optimization. However, these works provide useful tools solely for CAPA-based communications \cite{opt2,opt3,opt4,opt5,qian,chen2025implicit} or sensing \cite{chen2023cramer,chen2024near,jiang2024cram}, but do not directly resolve the coupled sensing–communication design.

Initial CAPA-ISAC works have started to characterize the system performance and investigate sensing-communications trade-offs in the continuous-aperture setting \cite{zhao2025downlink,zhangyue}. However, existing CAPA-ISAC studies predominantly consider deterministic line-of-sight (LoS) channels, which may be overly idealised and not reflect practical conditions. In contrast, CAPA-ISAC over fading channels remains largely unexplored. In realistic environments, multipath scattering and small-scale fading are ubiquitous and can fundamentally affect reliability (outage) and average performance (ergodic rate). Therefore, there remains a clear need for a unified and tractable analytical framework for CAPA-based ISAC over fading channels.

However, the extension from LoS to non-LoS models is nontrivial because it requires handling operator-based random channels. Moreover, EM randomness induced by scattering makes performance evaluation analytically challenging. In particular, under multipath scattering, the effective channel gain becomes a functional of a correlated continuous random field over the aperture, for which a closed-form distribution is generally unavailable. This prevents rigorous derivation of performance metrics, e.g., ergodic communication rate (ECR) and outage probability (OP), and tractable sensing–communication trade-off characterizations for CAPA-ISAC.

Motivated by the above gaps and challenges, this paper develops a unified modeling and analysis framework for CAPA-based ISAC. Specifically, the sensing link is modeled via spherical-wave propagation and a Swerling-I radar cross-section (RCS) model, while the communication link is modeled as an EM multipath channel induced by scatterers, leading to a continuous random-field representation over the transmit aperture. Our main contributions are summarized as follows:

\begin{itemize}
	\item We establish a continuous-aperture signal model for dual-functional sensing and
	communications (DFSAC) transmission where the beamformer is a normalized source-current field over the transmit aperture. MI-based sensing and communication performance metrics, i.e., sensing rate (SR) and communication rate (CR), are explicitly expressed as functionals of the continuous beamformer.

	\item By characterizing the autocorrelation structure of the aperture-domain fading channel, we connect the eigenvalue behaviour of the induced $\mathrm{sinc}$ kernel operator to the effective spatial DoFs. Using Landau’s eigenvalue theorem, we show an eigenvalue “polarization” phenomenon governed by DoFs, which enables accurate finite-sum approximations for the continuous channel gain distribution. This yields an analytically tractable pathway to evaluate ergodic and outage metrics in CAPA fading channels.
	
	\item We investigate two continuous beamforming strategies: the sensing-centric (S-C) design that maximizes SR, and the communication-centric (C-C) design that maximizes CR. For each design, we derive closed-form expressions for SR, ECR, and OP, together with high-SNR asymptotic results that explicitly reveal the high-SNR slopes and diversity orders. Further, to quantify the sensing–communication trade-off, we formulate the Pareto boundary characterization problem for the achievable SR–CR region. By using a subspace method, we convert the intractable functional optimization into a low-dimensional vector optimization, from which a closed-form Pareto-optimal beamformer is obtained.
	
	\item Numerical results validate the derived expressions and show that CAPA-ISAC can strictly enlarge the achievable SR–CR region compared with conventional SPDA-based ISAC and CAPA-based frequency-division sensing and communications (FDSAC). The results further highlight the role of effective DoFs in explaining the observed diversity advantages, and it is shown that the benefits of aperture enlargement depend critically on beamformer alignment with the target metric, thereby motivating the proposed Pareto trade-off design. 
\end{itemize}

The remainder of this paper is organized as follows. Section \ref{section_system} introduces the CAPA-ISAC system model, including the sensing and communication signal models and the corresponding metrics. Section \ref{section_pre} characterizes the statistical properties of the fading communication channel and derives the probability density function (PDF) of the channel gain. Section \ref{section_I_CSI} presents the performance analysis under the S-C, C-C, and Pareto-optimal beamforming designs, including closed-form SR/ECR/OP expressions and multiplexing/diversity characterizations. Section \ref{section_numerical} provides numerical results and comparisons with SPDA-ISAC and CAPA-FDSAC baselines. Section \ref{section_conclusion} concludes the paper. 
\subsubsection*{Notations}
Throughout this paper, scalars, vectors, and matrices are denoted by non-bold, bold lower-case, and bold upper-case letters, respectively. For the matrix $\mathbf{A}$, ${\mathbf{A}}^{\mathsf{T}}$, ${\mathbf{A}}^{*}$ and ${\mathbf{A}}^{\mathsf{H}}$ denote its transpose, conjugate, and conjugate transpose, respectively. The notations $\lvert a\rvert$ and $\lVert \mathbf{a} \rVert$ denote the magnitude and norm of scalar $a$ and vector $\mathbf{a}$, respectively. The identity matrix with dimensions $N\times N$ is represented by $\mathbf{I}_N$. The set $\mathbb{R}$ and $\mathbb{C}$ stand for the real and complex spaces, respectively, and notation $\mathbb{E}\{\cdot\}$ represents mathematical expectation. The MI between random variables $X$ and $Y$ conditioned on $Z$ is shown by $I\left(X;Y|Z\right)$. Finally, ${\mathcal{CN}}({\bm\mu},\mathbf{X})$ is used to denote the circularly-symmetric complex Gaussian distribution with mean $\bm\mu$ and covariance matrix $\mathbf{X}$, and $\mathsf{Exp}\left(\lambda \right) $ represents the exponential
distribution with rate parameter $\lambda$.

\begin{figure}[!t]
 \centering
\setlength{\abovecaptionskip}{0pt}
\includegraphics[height=0.32\textwidth]{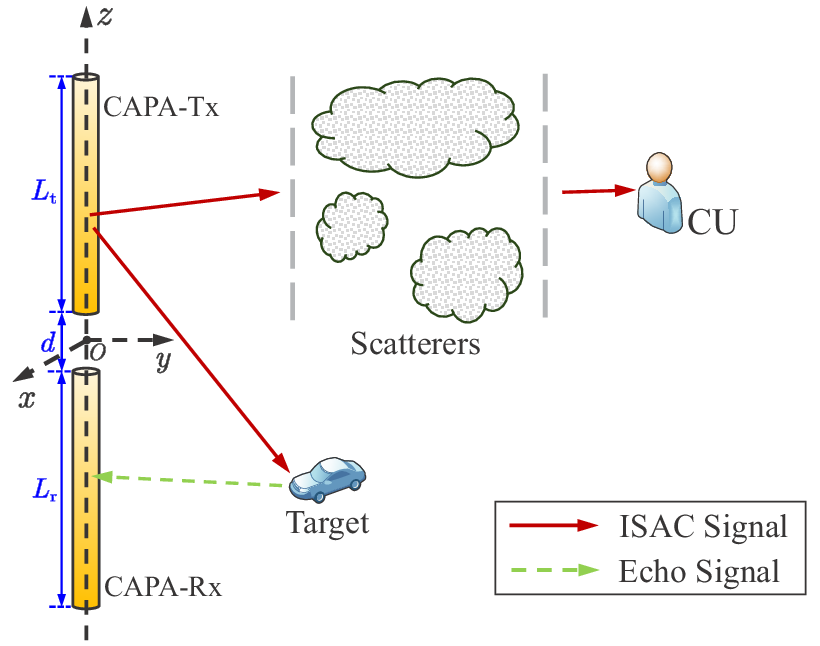}
\caption{Illustration of CAPA-based ISAC.}
\label{model_a}
\vspace{-3pt}
\end{figure}

\section{System Model}\label{section_system}
\subsection{System Description}
We consider a downlink CAPA-based ISAC system, where a DFSAC base station (BS) sends data to a single-antenna {communication user (CU)} while simultaneously sensing a point-like target. As illustrated in {\figurename} {\ref{model_a}}, the BS is equipped with two uni-polarized linear CAPAs\footnote{In this paper, we focus on one-dimensional linear CAPAs for notational simplicity and clear physical insight, while the framework and analysis naturally extend to two-dimensional planar apertures. Note that a linear CAPA is conceptually different from a conventional wire/line antenna, which is typically a resonant radiator whose current distribution is determined by its geometry and feed, rather than by independent point-wise control along the aperture.}: a transmit CAPA with a length of $L_{\rm{t}}$ and a receive CAPA with a physical length of $L_{\rm{r}}$. The CAPAs are positioned along the $z$-axis, separated end-to-end by a distance $d$, and the midpoint of this distance is defined as the origin. Therefore, the apertures of the transmit and receive CAPAs are, respectively, given by
\begin{subequations}
	\begin{align}
		&\mathcal{A} _{\mathrm{t}}=\{[0,0,z]^{\mathsf{T}}|z\in [{d}/{2},{d}/{2}+L_{\mathrm{t}}]\},\\
		&\mathcal{A} _{\mathrm{r}}=\{[0,0,z]^{\mathsf{T}}|z\in [-{d}/{2}-L_{\mathrm{r}},-{d}/{2}]\}.
	\end{align}
\end{subequations}

Let ${\mathbf{x}}({\mathbf{t}})=[x_1({\mathbf{t}}),\ldots,x_L({\mathbf{t}})]\in{\mathbb{C}}^{1\times L}$ represent the DFSAC signal emitted by the transmit CAPA at ${\mathbf{t}}\in{\mathcal{A}}_{\rm{t}}$, where $L$ denotes the number of symbols in a communication frame or the count of sensing pulses. Within the communication context, each $x_{\ell}({\mathbf{t}})$ corresponds to the $\ell$th transmitted symbol, for $\ell = 1, \ldots, L$. In the sensing scenario, $x_\ell({\mathbf{t}})$ indicates the sensing snapshot transmitted at the $\ell$th time interval. {Throughout this paper, we adopt a frame-based quasi-static model. Each DFSAC frame contains $L$ symbols/pulses and is assumed to be shorter than the effective coherence intervals of both the sensing and communication links. Accordingly, the target RCS-related coefficient and the communication channel remain fixed within each frame, but may vary independently across frames.}

In the CAPA-ISAC framework, the continuous DFSAC signal is constructed as follows:
\begin{align}
{\mathbf{x}}({\mathbf{t}})=\sqrt{P}w({\mathbf{t}}){{\mathbf{s}}^{\mathsf{T}}},
\end{align}
where $P$ is the transmit power, ${\mathbf{s}}\in{\mathbb{C}}^{L\times1}$ is a unit-power data stream intended for the CU, and $w({\mathbf{t}})$ denotes the normalized source current or transmit beamformer, which maps each spatial point $\mathbf{t}$ on the transmitting array to a complex scalar. Here, $w(\cdot)$ is a scalar field defined over the transmit aperture, and the beamformer satisfies the energy constraint:
\begin{align}
\int_{{\mathcal{A}}_{\rm{t}}}\lvert{w}({\mathbf{t}})\rvert^2{\rm{d}}{\mathbf{t}}=1. 
\end{align}
The data vector $\mathbf{s} = [s_1, \ldots, s_L]^{\mathsf{T}}$ fulfills the condition $\frac{1}{L} | \mathbf{s} |^2 = 1$, where each $s_\ell$ represents a modulated symbol sent at the $\ell$th interval. Consequently, the transmitted signal at time index $\ell$ becomes 
\begin{align}\label{x_l}
{{x}}_\ell({\mathbf{t}})=\sqrt{P}s_\ell{{w}}({\mathbf{t}}),\quad \ell=1,\ldots,L.
\end{align}

\subsection{Sensing Model}
At the $\ell$th time interval, the received echo signal at the BS at point $\mathbf{r}\in{\mathcal{A}_{\rm{r}}}$ for target sensing can be written as follows:
\begin{align}\label{sensing_Model_basic_one}
	y_{{\rm{s}},\ell}(\mathbf{r})=\int_{{\mathcal{A}}_{\rm{t}}}q(\mathbf{r},{\mathbf{t}})x_\ell({\mathbf{t}}){\rm{d}}{\mathbf{t}}
	+n_{{\rm{s}},\ell}(\mathbf{r}),
\end{align}
where $y_{{\rm{s}},\ell}(\cdot)$ is the receive vector field, $q(\cdot,\cdot)$ is a complex integral kernel denoting the sensing response, and $n_{{\rm{s}},\ell}(\cdot)$ represents thermal noise. The term $q(\mathbf{r},{\mathbf{t}})$ gives the sensing response between the transmit position ${\mathbf{t}}\in{\mathcal{A}}_{\rm{t}}$ and the receive position $\mathbf{r}\in{\mathcal{A}_{\rm{r}}}$. The noise field is characterized as a zero-mean complex Gaussian random process, satisfying 
\begin{align}
	{\mathbb{E}}\{n_{{\rm{s}},\ell}(\mathbf{r})n_{{\rm{s}},\ell}^*(\mathbf{r}')\}={\sigma}_{\rm{s}}^2\delta(\mathbf{r}-{\mathbf{r}}'), 
\end{align}
where $\delta(\cdot)$ is the Dirac delta function, and ${\sigma}_{\rm{s}}^2$ is the power spectral density of the noise. It is also assumed that the transmitted symbol sequence $\mathbf{s}$ and the noise fields $\{n_{{\rm{s}},\ell}(\mathbf{r})\}_{\ell=1}^{L}$ are statistically independent.

The sensing response can be modeled as follows \cite{pizzo2022spatial}:
\begin{align}\label{Target_Response_General}
	q(\mathbf{r},{\mathbf{t}})=\int h_{\rm{s}}(\mathbf{r},\mathbf{p})a(\mathbf{p})h_{\rm{s}}(\mathbf{p},\mathbf{t}){\rm{d}}{\mathbf{p}},
\end{align}
where $h_{\rm{s}}(\mathbf{x},\mathbf{y})$ denotes the sensing link between points $\mathbf{x}$ and $\mathbf{y}$, and $a(\mathbf{p})$ corresponds to the scattering coefficient or reflectivity, commonly referred to as the RCS, at point $\mathbf{p}$. To precisely capture the characteristics of the LoS sensing links in both near-field and far-field scenarios, we model the sensing links based on the non-uniform spherical-wave propagation as follows \cite{ouyang2024impact}:
\begin{align}
	h_{\rm{s}}(\mathbf{x},\mathbf{y})=\frac{{\rm{e}}^{-{\rm{j}}k_0\lVert\mathbf{x}-\mathbf{y}\rVert}}{\sqrt{4\pi}\lVert\mathbf{x}-\mathbf{y}\rVert},
\end{align} 
where $k_0=\frac{2\pi}{\lambda}$ with $\lambda$ being the wavelength denotes the wavenumber.

Based on the ray-tracing model in \cite{zwick2002stochastic}, the scattering response can be expressed as follows:
\begin{align}
	a(\mathbf{p})=\sum\nolimits_{n=1}^{N}\beta_n\delta({\mathbf{p}}-{\mathbf{p}}_n),
\end{align}
where $N$ is the number of resolvable point-like targets, ${\mathbf{p}}_n\in{\mathbb{R}}^{3\times1}$ represents the location of the $n$th target, and $\beta_n$ denotes the associated RCS. { In this work, we consider a single point-like target ($N=1$) to obtain tractable closed-form analysis and clear physical insights, while leaving the extension to multiple targets for future work.} For a target located at
\begin{equation}
	{\mathbf{p}}_{\rm{s}}=[p_{\rm{s},x},p_{\rm{s},y},p_{\rm{s},z}]^\mathsf{T},
\end{equation} 
with RCS $\beta_{\rm{s}}$, the scattering response reduces to
\begin{align}
	a(\mathbf{p})=\beta_{\rm{s}}\delta({\mathbf{p}}-{\mathbf{p}}_{\rm{s}}).
\end{align}
{Substituting into \eqref{Target_Response_General}, the response becomes:
\begin{align}
	q(\mathbf{r},{\mathbf{t}})=h_{\rm{s}}(\mathbf{r},{\mathbf{p}}_{\rm{s}})\beta_{\rm{s}}h_{\rm{s}}({\mathbf{p}}_{\rm{s}},\mathbf{t}).
\end{align}
For notational simplicity, we denote $h_{\rm{s}}(\mathbf{p}_{\rm{s}}, \mathbf{t})\triangleq h_{\rm{s}}(\mathbf{t})$ for $\mathbf{t}\in \mathcal{A} _{\mathrm{t}}$ and $h_{\rm{s}}(\mathbf{r}, \mathbf{p}_{\rm{s}})\triangleq h_{\rm{s}}(\mathbf{r})$ for $\mathbf{r}\in \mathcal{A} _{\mathrm{r}}$. Accordingly, the received echo signal at point $\mathbf{r}$ can be expressed as follows:
\begin{align}
	y_{{\rm{s}},\ell}(\mathbf{r})=h_{\rm{s}}(\mathbf{r})\beta_{\rm{s}}\int_{{\mathcal{A}}_{\rm{t}}}h_{\rm{s}}(\mathbf{t})
	x_\ell({\mathbf{t}}){\rm{d}}{\mathbf{t}}
	+n_{{\rm{s}},\ell}(\mathbf{r}).
\end{align}}
Following the Swerling-I model \cite{richards2005fundamentals}, the RCS $\beta_{\rm{s}}$ is considered stable across different pulses and is characterized by a Rayleigh-distributed magnitude. Consequently, $\beta_{\rm{s}}$ is modeled as a circularly symmetric complex Gaussian random variable, i.e., $\beta_{\rm{s}} \sim \mathcal{CN}(0, \alpha_{\rm{s}})$, where $\alpha_{\rm{s}} > 0$ denotes the mean power of the RCS. This quantity reflects the average strength of the target’s reflectivity \cite{richards2005fundamentals}.

{ In this work, we assume that the target position is known, or has been estimated through a prior acquisition or tracking stage, and focus on estimating the RCS $\beta_{\rm{s}}$ associated with the direction of interest. This coefficient encapsulates the intrinsic scattering properties of the target, e.g., material, size, shape, and orientation, which are crucial for downstream sensing tasks such as target detection, recognition, and tracking \cite{ouyang2023integrated,jiang2024electromagnetic}. Moreover, treating $\beta_{\rm{s}}$ as the unknown information-bearing parameter casts the sensing model as a virtual Multiple-input single-output (MISO) channel. This allows the sensing quality to be quantified by the MI between $\beta_{\rm{s}}$ and the echo signals \cite{ouyang2023integrated,bcrb}, under the same information-theoretic measure as the CR. This is essential for establishing a unified sensing-communication trade-off, which would be difficult to obtain if geometric parameters such as the position are estimated, since their accuracy is naturally characterized by the Cram\'er-Rao bound rather than an MI-based rate. Such a position-known, RCS-unknown formulation has also been widely adopted in radar and ISAC studies \cite{boqun_jstsp,liu2022cramer,dong2023joint,tang2019spectrally}. Extending the framework to joint position and reflection-coefficient estimation remains an important direction for future work.}

In this context, the objective of sensing is to retrieve environmental information encoded in $\beta_{\rm{s}}$ from the observed echo signals $\{y_{{\rm{s}},\ell}(\mathbf{r})\}_{\ell=1}^{L}$, utilizing prior knowledge of the transmitted signal ${\mathbf{x}}({\mathbf{t}})$. To this end, the receiver employs a detector ${{v}}_{\rm{s}}({\mathbf{r}})$ for $\mathbf{r}\in\mathcal{A}_{\rm{r}}$, which yields
\begin{equation}\label{sensing_1}
	\begin{split}
		\int_{{\mathcal{A}}_{\rm{r}}}{}{v}_{\rm{s}}^{*}({\mathbf{r}}){{y}}_{{\rm{s}},\ell}(\mathbf{r}){\rm{d}}{\mathbf{r}}=\sqrt{P}s_{\ell}\beta_{\rm{s}}\int_{{\mathcal{A}}_{\rm{r}}}{{v}}_{\rm{s}}^{*}({\mathbf{r}})h_{\rm{s}}(\mathbf{r}){\rm{d}}{\mathbf{r}}&\\
		\times\int_{{\mathcal{A}}_{\rm{t}}}h_{\rm{s}}(\mathbf{t})w({\mathbf{t}}){\rm{d}}{\mathbf{t}}
		+\int_{{\mathcal{A}}_{\rm{r}}}{{v}}_{\rm{s}}^{*}({\mathbf{r}}){{n}}_{{\rm{s}},\ell}(\mathbf{r}){\rm{d}}{\mathbf{r}}&.
	\end{split}
\end{equation}
For clarity, we denote $\hat{y}_{{\rm{s}},\ell}\triangleq\int_{{\mathcal{A}}_{\rm{r}}}{{v}}_{\rm{s}}^{*}({\mathbf{r}}){{y}}_{{\rm{s}},\ell}(\mathbf{r}){\rm{d}}{\mathbf{r}}$, $\hat{{h}}_{\rm{r}}\triangleq\int_{{\mathcal{A}}_{\rm{r}}}{{v}}_{\rm{s}}^{*}({\mathbf{r}})h_{\rm{s}}(\mathbf{r}){\rm{d}}{\mathbf{r}}$, $\hat{{h}}_{\rm{t}}\triangleq\int_{{\mathcal{A}}_{\rm{t}}}h_{\rm{s}}(\mathbf{t})w({\mathbf{t}}){\rm{d}}{\mathbf{t}}$, and  $\hat{n}_{{\rm{s}},\ell}\triangleq\int_{{\mathcal{A}}_{\rm{r}}}{{v}}_{\rm{s}}^{*}({\mathbf{r}}){{n}}_{{\rm{s}},\ell}(\mathbf{r}){\rm{d}}{\mathbf{r}}$. As per \cite[Lemma 1]{zhao2024continuous}, we have $\hat{n}_{{\rm{s}},\ell}\sim{\mathcal{CN}}(0,{\sigma}_{\rm{s}}^2\int_{\mathcal{A}_{\rm{r}}}\lvert{{v}}_{\rm{s}}({\mathbf{r}})\rvert^2{\rm{d}}{\mathbf{r}})$. Hence, \eqref{sensing_1} can be rewritten as follows:
\begin{align}
	\hat{y}_{{\rm{s}},\ell}=\sqrt{P}\hat{{h}}_{\rm{r}}\beta_{\rm{s}}\hat{{h}}_{\rm{t}}s_{\ell}+\hat{n}_{{\rm{s}},\ell},~\ell=1,\ldots,L,
\end{align}
which yields
\begin{align}\label{Sensing_Model_Transformed}
	\hat{\mathbf{y}}_{{\rm{s}}}=\sqrt{P}\hat{{h}}_{\rm{r}}\hat{{h}}_{\rm{t}}{\mathbf{s}}\beta_{\rm{s}}+\hat{\mathbf{n}}_{{\rm{s}}},
\end{align}
where ${\hat{\mathbf{y}}_{{\rm{s}}}}=[\hat{y}_{{\rm{s}},1},\ldots,\hat{y}_{{\rm{s}},L}]^{\mathsf{T}}\in{\mathbb{C}}^{L\times1}$ and ${\hat{\mathbf{n}}_{{\rm{s}}}}=[\hat{n}_{{\rm{s}},1},\ldots,\hat{n}_{{\rm{s}},L}]^{\mathsf{T}}\sim{\mathcal{CN}}({\mathbf{0}},{\sigma}_{\rm{s}}^2\int_{{\mathcal{A}}_{\rm{r}}}\lvert{{v}}_{\rm{s}}({\mathbf{r}})\rvert^2{\rm{d}}{\mathbf{r}}{\mathbf{I}}_L)$. 

The next task is to recover $\beta_{\rm{s}}$ from $\hat{\mathbf{y}}_{{\rm{s}}}$ given ${\mathbf{s}}$. The information-theoretic limits of sensing performance are characterized by the sensing MI, which quantifies how much information about $\beta_{\rm{s}}$ is contained in $\hat{\mathbf{y}}_{\rm{s}}$ when conditioned on $\mathbf{s}$ \cite{ouyang2023integrated}. To assess performance, we adopt the SR as the performance metric, which is defined as the sensing MI per unit time. Notably, under our considered model, a larger SR also implies a smaller mean-squared error and Bayesian Cram{\'e}r-Rao bound \cite{bcrb,zhao2025downlink}. { Since each DFSAC symbol lasts one unit of time, the SR can be written as follows:
\begin{align}\label{SR_define}
	{\mathcal{R}}_{\rm{s}}({{w}}({\mathbf{t}}))=\frac{1}{L}I(\hat{\mathbf{y}}_{{\rm{s}}};\beta_{\rm{s}}|{\mathbf{s}}),
\end{align}
where $I(\hat{\mathbf{y}}_{{\rm{s}}};\beta_{\rm{s}}|{\mathbf{s}})$ denotes the sensing MI accumulated over the $L$-symbol frame. The factor $1/L$ normalizes this total sensing MI into an information rate per unit time, placing the SR on the same per-channel-use basis as the CR \cite{ouyang2023integrated}.} Viewing the model in \eqref{Sensing_Model_Transformed} as a virtual single-input multiple-output channel, we interpret $\beta_{\rm{s}}$ as the input, $\hat{\mathbf{y}}_{\rm{s}}$ as the output, and $\hat{h}_{\rm{r}} \hat{h}_{\rm{t}} \mathbf{s}$ as the effective channel gain, with an additive noise $\hat{\mathbf{n}}_{\rm{s}}$. The resulting MI is equivalent to the capacity of this virtual MISO channel, computed as
\begin{align}\label{MI}
	I(\hat{\mathbf{y}}_{{\rm{s}}};\beta_{\rm{s}}|{\mathbf{s}})
	&=\log_2\det\left({\mathbf{I}}_L+\frac{P\alpha_{\rm{s}}\lvert\hat{{h}}_{\rm{r}}\rvert^2\lvert\hat{{h}}_{\rm{t}}\rvert^2}{{\sigma}_{\rm{s}}^2
		\int_{{\mathcal{A}}_{\rm{r}}}\lvert{{v}}_{\rm{s}}({\mathbf{r}})\rvert^2{\rm{d}}{\mathbf{r}}}
	{\mathbf{s}}{\mathbf{s}}^{\mathsf{H}}\right)\notag\\
	&=\log_2\left(1+\frac{P\alpha_{\rm{s}}\lvert\hat{{h}}_{\rm{r}}\rvert^2\lvert\hat{{h}}_{\rm{t}}\rvert^2}{{\sigma}_{\rm{s}}^2
		\int_{{\mathcal{A}}_{\rm{r}}}\lvert{{v}}_{\rm{s}}({\mathbf{r}})\rvert^2{\rm{d}}{\mathbf{r}}}
	\lVert{\mathbf{s}}\rVert^2\right),
\end{align}
where the second equality results from applying Sylvester's determinant identity. Substituting \eqref{MI} into \eqref{SR_define} and using $\frac{1}{L}\lVert{\mathbf{s}}\rVert^2=1$, the SR can be expressed as follows:
\begin{align}\label{SR_define_2}
	{\mathcal{R}}_{\rm{s}}({{w}}({\mathbf{t}}))=\frac{1}{L}\log_2\left(1+\frac{P\alpha_{\rm{s}}\lvert\hat{{h}}_{\rm{r}}\rvert^2\lvert\hat{{h}}_{\rm{t}}\rvert^2}{{\sigma}_{\rm{s}}^2
		\int_{\mathcal{A}_{\rm{r}}}\lvert{{v}}_{\rm{s}}({\mathbf{r}})\rvert^2{\rm{d}}{\mathbf{r}}}L\right).
\end{align}

Here, SR is influenced by the choice of detector $v_{\rm{s}}(\mathbf{r})$. To maximize the SR, we solve the following optimization problem:
\begin{align}\label{Downlink_SR_Pre_Opt1}
	\max_{{{v}}_{\rm{s}}({\mathbf{r}})}\frac{\lvert\hat{{h}}_{\rm{r}}\rvert^2}{\int_{\mathcal{A}_{\rm{r}}}\lvert{{v}}_{\rm{s}}({\mathbf{r}})\rvert^2{\rm{d}}{\mathbf{r}}}\Leftrightarrow
	\max_{{{v}}_{\rm{s}}({\mathbf{r}})}\frac{\left\lvert\int_{{\mathcal{A}}_{\rm{r}}}{{v}}_{\rm{s}}^{*}({\mathbf{r}})h_{\rm{s}}(\mathbf{r}){\rm{d}}{\mathbf{r}}\right\rvert^2}{\int_{\mathcal{A}_{\rm{r}}}\lvert{{v}}_{\rm{s}}({\mathbf{r}})\rvert^2{\rm{d}}{\mathbf{r}}}.
\end{align}
The optimal solution occurs when $v_{\rm{s}}(\mathbf{r})$ is proportional to $h_{\rm{s}}(\mathbf{r})$, i.e., $v_{\rm{s}}(\mathbf{r}) \propto h_{\rm{s}}(\mathbf{r})$, which gives
\begin{align}\label{Downlink_SR_Pre_Opt2}
	\max_{{{v}}_{\rm{s}}({\mathbf{r}})}\frac{\lvert\hat{{h}}_{\rm{r}}\rvert^2}{\int_{\mathcal{A}_{\rm{r}}}\lvert{{v}}_{\rm{s}}({\mathbf{r}})\rvert^2{\rm{d}}{\mathbf{r}}}
	=\frac{\lvert\int_{\mathcal{A}_{\rm{r}}}\lvert h_{\rm{s}}(\mathbf{r})\rvert^2{\rm{d}}{\mathbf{r}}\rvert^2}{\int_{\mathcal{A}_{\rm{r}}}\lvert h_{\rm{s}}(\mathbf{r})\rvert^2{\rm{d}}{\mathbf{r}}}=\int_{\mathcal{A}_{\rm{r}}}\lvert h_{\rm{s}}(\mathbf{r})\rvert^2{\rm{d}}{\mathbf{r}}. 
\end{align}
Finally, substituting \eqref{Downlink_SR_Pre_Opt2} back into \eqref{SR_define_2} yields
\begin{equation}\label{SR_define_2_step1}
	\begin{split}
		{\mathcal{R}}_{\rm{s}}({{w}}({\mathbf{t}}))=\frac{1}{L}\log_2&\bigg(1+\overline{\gamma }_{\mathrm{s}}L\alpha_{\mathrm{s}}\int_{\mathcal{A}_{\rm{r}}}\lvert h_{\rm{s}}(\mathbf{r})\rvert^2{\rm{d}}{\mathbf{r}}\\
		&\times\left| \int_{\mathcal{A} _{\mathrm{t}}}{h_{\mathrm{s}}}(\mathbf{t})w(\mathbf{t})\mathrm{d}\mathbf{t}  \right|^2\Bigg),
	\end{split}	
\end{equation}
where $\overline{\gamma }_{\mathrm{s}}\triangleq\frac{P}{\sigma _{\mathrm{s}}^{2}}$.

\subsection{Communication Model}
At the $\ell$th time interval, the received signal at the CU is given by
\begin{align}\label{comm_signal}
y_{{\rm{c}},\ell}=\int_{{\mathcal{A}}_{\rm{t}}}h_{\rm{c}}(\mathbf{p}_{\rm{c}},{\mathbf{t}})x_\ell({\mathbf{t}}){\rm{d}}{\mathbf{t}}
+n_{{\rm{c}},\ell},
\end{align}
for $\ell=1,\ldots,L$, where $\mathbf{p}_{\rm{c}}$ is the position of the CU,  $h_{\rm{c}}(\mathbf{p}_{\rm{c}},{\mathbf{t}})$ represents the spatial response from ${\mathbf{t}}\in{\mathcal{A}}_{\rm{t}}$ to the CU, and $n_{{\rm{c}},\ell}\sim \mathcal{C} \mathcal{N} \left( 0, {\sigma}_{\rm{c}}^2\right) $ represents the additive noise. 
 
We consider a scenario in which no direct path exists between the transmit CAPA and the CU due to the presence of scatterers. In this case, the EM multipath spatial response, $h_{\rm{c}}(\mathbf{p}_{\rm{c}},{\mathbf{t}})\triangleq h_{\rm{c}}({\mathbf{t}})$, can be modeled as follows \cite{pizzo2022spatial}:
\begin{align}\label{comm_model}
	h_{\mathrm{c}}(\mathbf{t})=\frac{1}{(2\pi) ^2}\iiiint_{\mathcal{D} \left( \mathbf{k} \right) \times \mathcal{D} \left( \bm{\kappa } \right)}{\mathrm{e}^{-\mathrm{j}\mathbf{k}^{\mathsf{T}}\mathbf{p}_{\mathrm{c}}}H\left( \mathbf{k},\bm{\kappa } \right) \mathrm{e}^{\mathrm{j}\bm{\kappa }^{\mathsf{T}}\mathbf{t}}}\mathrm{d}\mathbf{k}\mathrm{d}\bm{\kappa },
\end{align}
where $\mathbf{k}=\left[ k_x,\gamma \left( k_x,k_z \right) ,k_z \right]^{\mathsf{T}} \in \mathbb{R} ^{3\times 1}$, $\boldsymbol{\kappa }=\left[ \kappa _x,\gamma \left( \kappa _x,\kappa _z \right) ,\kappa _z \right]^{\mathsf{T}} \in \mathbb{R} ^{3\times 1}$, $\gamma \left( x,y \right) =\sqrt{k_{0}^{2}-x^2-y^2}$, $\mathcal{D} \left( \mathbf{k} \right) =\left\{ \left( k_x,k_z \right) \left| k_{x}^{2}+k_{z}^{2}\le k_{0}^{2} \right. \right\} $, and $\mathcal{D} \left( \boldsymbol{\kappa } \right) =\left\{ \left( \kappa _x,\kappa _z \right) \left| \kappa _{x}^{2}+\kappa _{z}^{2}\le \kappa _{0}^{2} \right. \right\} $\footnote{ The communication and sensing channels stem from the same EM propagation principle, while their different forms reflect distinct physical scenarios. For sensing, the link is a round-trip BS-target-BS path with a known target location. Any additional scatterer along this round trip incurs multiple path losses and returns negligible power, so the sensing link is dominated by the single target reflection and the two BS-target LoS paths, which justifies the deterministic spherical-wave model. The remaining non-negligible scatterer contribution is single-bounce BS-scatterer-BS clutter, which is treated as interference and assumed to be estimated and suppressed before sensing. By contrast, the communication link is modeled as a stochastic multipath field because the CU is considered in rich scattering without a dominant LoS path. This model preserves the fading behavior that the paper aims to analyze and enables key performance insights, such as the diversity order and high-SNR slope. This deterministic-sensing and stochastic-communication modeling approach is also widely adopted in the ISAC literature \cite{gan2024coverage,wei2025isac,bazzi2023outage,holoISAC_2,meng2025cooperative}.}. In particular, $\mathrm{e}^{\mathrm{j}\boldsymbol{\kappa }^{\mathsf{T}}\mathbf{t}}$ and $\mathrm{e}^{-\mathrm{j}\mathbf{k}^{\mathsf{T}}\mathbf{p}_{\mathrm{c}}}$ denote the transmit and receive plane waves propagating in directions $\frac{\boldsymbol{\kappa }}{\left\| \boldsymbol{\kappa } \right\|}$ and $\frac{\mathbf{k}}{\left\| \mathbf{k} \right\|}$, respectively, and $H\left( \mathbf{k},\boldsymbol{\kappa } \right) $ represents the angular response that maps the transmit direction $\frac{\boldsymbol{\kappa }}{\left\| \boldsymbol{\kappa } \right\|}$ to the receive direction $\frac{\mathbf{k}}{\left\| \mathbf{k} \right\|}$. For analytical tractability, by considering the isotropic scattering, $H\left( \mathbf{k},\boldsymbol{\kappa } \right) $ is modeled as a stationary, circularly symmetric, complex-Gaussian random field, which yields
\begin{align}\label{H_define}
	H\left( \mathbf{k},\boldsymbol{\kappa } \right) =\frac{A_s( k_0 )}{\sqrt{\gamma \left( k_x,k_z \right) \gamma \left( \kappa _x,\kappa _z \right)}}W\left( \mathbf{k},\boldsymbol{\kappa } \right). 
\end{align}
Here, $W\left( \mathbf{k},\boldsymbol{\kappa } \right)$ is a zero-mean, unit-variance complex-Gaussian (ZUCG) random field defined over $\mathcal{D} \left( \mathbf{k} \right) \times \mathcal{D} \left( \mathbf{\kappa } \right) $, i.e., $W\left( \mathbf{k},\boldsymbol{\kappa } \right) \sim \mathcal{C} \mathcal{N} \left( 0,1 \right) $ and its autocorrelation satisfies $\mathbb{E}\left\{ W\left( \mathbf{k},\boldsymbol{\kappa } \right) W^*\left( \mathbf{k}^{\prime},\boldsymbol{\kappa }^{\prime} \right) \right\} =\delta \left( \mathbf{k}-\mathbf{k}^{\prime} \right) \delta \left( \boldsymbol{\kappa }-\boldsymbol{\kappa }^{\prime} \right) $, and $A_s\left( k_0 \right)=\frac{2\pi}{k_0}$. {Consistent with the frame-based quasi-static model, one realization of the communication fading field $h_{\rm{c}}(\mathbf{t})$ remains fixed over the $L$ symbols within one frame.}

Substituting \eqref{x_l} into \eqref{comm_signal}, we have
\begin{align}
	{{y}}_{{\rm{c}},\ell}= \sqrt{P}s_\ell\int_{{\mathcal{A}}_{\rm{t}}}h_{\rm{c}}({\mathbf{t}})
	w({\mathbf{t}}){\rm{d}}{\mathbf{t}}
	+n_{{\rm{c}},\ell}.
\end{align}
Therefore, the CR is expressed as follows:
\begin{align}\label{CR_define}
	{\mathcal{R}}_{{\rm{c}}}({{w}}({\mathbf{t}}))=\log_2\left(1+\overline{\gamma}_{\mathrm{c}} \left| \int_{\mathcal{A} _{\mathrm{t}}}{h_{\mathrm{c}}}(\mathbf{t})w(\mathbf{t})\mathrm{d}\mathbf{t}  \right|^2\right),
\end{align} 
where $\overline{\gamma}_{\mathrm{c}}\triangleq\frac{P}{\sigma _{\mathrm{c}}^{2}}$.

{In the considered CAPA-based ISAC system, sensing and communication are integrated through shared physical and signal resources. Specifically, the two functions use the same transmit CAPA, spectrum, power budget, and waveform. The communication data stream $\mathbf{s}$ is reused as the sensing probing sequence, while a common source current pattern $w(\mathbf{t})$ shapes the EM field toward both the CU and the target. Hence, the integration of sensing and communication is governed by the design of this shared continuous beamformer.}

By observing \eqref{CR_define} and \eqref{SR_define_2_step1}, we note that both communication and sensing performance in CAPA-based ISAC systems heavily depend on $w(\mathbf{t})$. Nevertheless, identifying an optimal $w(\mathbf{t})$ to simultaneously maximize both the CR and SR is impossible. To tackle this issue and examine the fundamental performance limits of CAPA-based ISAC systems, we introduce three distinct continuous beamforming strategies: \romannumeral1) \emph{the S-C design} that maximizes ${\mathcal{R}}_{\rm{s}}({{w}}({\mathbf{t}}))$; \romannumeral2) \emph{the C-C design} that maximizes ${\mathcal{R}}_{{\rm{c}}}({{w}}({\mathbf{t}}))$; \romannumeral3) \emph{the Pareto-optimal design} that identifies the Pareto boundary of the SR-CR region.

\section{Communication Channel Statistics}\label{section_pre}
In this section, we characterize the statistical properties for the communication channel.

By substituting \eqref{H_define} into \eqref{comm_model}, the communication channel response can be expressed as follows:
\begin{align}\label{comm_model_2}
h_{\mathrm{c}}(\mathbf{t})=\frac{A_s\left( k_0 \right)}{4\pi ^2}\iint_{\mathcal{D} \left( \mathbf{\kappa } \right)}{\frac{\mathrm{e}^{\mathrm{j}\boldsymbol{\kappa }^{\mathsf{T}}\mathbf{t}}}{\sqrt{\gamma \left( \kappa _x,\kappa _z \right)}}}\hat{H}\left( \boldsymbol{\kappa } \right) \mathrm{d}\mathbf{\kappa },
\end{align}
where 
\begin{align}
	\hat{H}\left( \boldsymbol{\kappa } \right) =\iint_{\mathcal{D} \left( \mathbf{k} \right)}{\frac{\mathrm{e}^{-\mathrm{j}\mathbf{k}^{\mathsf{T}}\mathbf{p}_{\mathrm{c}}}}{\sqrt{\gamma \left( k_x,k_z \right)}}W\left( \mathbf{k},\boldsymbol{\kappa } \right) }\mathrm{d}\mathbf{k}.
\end{align}
Since $W\left( \mathbf{k},\boldsymbol{\kappa } \right) $ is a ZUCG random field, we can conclude that $\hat{H}\left( \boldsymbol{\kappa } \right)$ is also a zero-mean complex-Gaussian random field. Given that $\mathbf{t}=[0,0,z]^\mathsf{T}$ for ${z}\in \mathcal{Z} \triangleq \left\{ \bar{z}\in \mathbb{R}| d/2\le \bar{z}\le d/2+L_{\mathrm{t}} \right\} $ and $\boldsymbol{\kappa }=\left[ \kappa _x,\gamma \left( \kappa _x,\kappa _z \right) ,\kappa _z \right]^{\mathsf{T}}$, \eqref{comm_model_2} can be further simplified as follows:
\begin{align}
	h_{\mathrm{c}}(\mathbf{t})=\frac{A_s\left( k_0 \right)}{4\pi ^2}\int_{-k_0}^{k_0}{\mathrm{e}^{\mathrm{j}\kappa _zz}\hat{H}_z\left( \kappa _z \right)}\mathrm{d}\kappa _z\triangleq \hat{h}_{\rm{c}}\left( z \right),
\end{align}
where
\begin{align}
	\hat{H}_z\left( \kappa _z \right) =\int_{-\sqrt{k_{0}^{2}-\kappa _{z}^{2}}}^{\sqrt{k_{0}^{2}-\kappa _{z}^{2}}}{\frac{\hat{H}\left( \boldsymbol{\kappa } \right)}{\sqrt{\gamma \left( \kappa _x,\kappa _z \right)}}}\mathrm{d}\kappa _x
\end{align}
is a zero-mean complex-Gaussian random field defined on $[-k_0,k_0]$. Therefore, the discussion above implies that the communication channel $\hat{h}_{\rm{c}}\left( z \right)$ is a zero-mean complex-Gaussian random field defined on $\mathcal{Z}$, whose statistics are determined by its autocorrelation function given as follows.
\begin{lemma}\label{lem_Rz}
The autocorrelation function of $\hat{h}_{\rm{c}}( z )$ is given by
\begin{align}\label{Rc}
	R_{\rm{c}}( z,z^{\prime} ) =\frac{1}{2k_0}\int_{-k_0}^{k_0}{\mathrm{e}^{\mathrm{j}\kappa _z\left( z-z^{\prime} \right)}\mathrm{d}\kappa _z}=\frac{\sin \left( k_0(z-z^\prime) \right)}{k_0(z-z^\prime)}.
\end{align}
\end{lemma}
\vspace{-3pt}
\begin{IEEEproof}
Please refer to Appendix \ref{proof_lem_Rz} for more details.
\end{IEEEproof}
Given that $R_{\rm{c}}( z,z^{\prime} )$ is a semi-positive definite Hilbert–Schmidt operator, we denote its eigendecomposition (EVD) of as follows:
\begin{align}
	R_{\rm{c}}( z,z^{\prime} ) =\sum\nolimits_{n=1}^{\infty}{\lambda _n}\phi _n\left( z \right) \phi _{n}^{*}\left( z^{\prime} \right), 
\end{align}
where $\lambda _1\ge \lambda _2\ge ...\ge \lambda _{\infty}\ge 0$ are the eigenvalues, and $\phi _n\left( z \right)$ is the associate eigenfunction satisfying
\begin{align}
	\int_{\mathcal{Z}}{\phi _n\left( z \right) \phi _{n^{\prime}}^{*}\left( z \right)}\mathrm{d}z=\delta _{n,n^{\prime}}
\end{align}
for $n\in[1,2,\ldots,\infty]$ with $\delta _{n,n^{\prime}}$ being the Kronecker delta. The eigenvalues can be evaluated numerically as outlined in the Appendix~\ref{eigenvalue}. On this basis, we can derive the following lemma.
\begin{lemma}\label{lem_h_sta}
The statistics of the zero-mean complex-Gaussian random field $\hat{h}_{\rm{c}}\left( z \right)$ satisfy
\begin{align}\label{hc_sta}
	\hat{h}_{\rm{c}}\left( z \right)\overset{d}{=}\sum\nolimits_{n=1}^{\infty}{\sqrt{\lambda _n}}\phi _n\left( z \right) \Psi  _n,
\end{align}
where $\overset{d}{=}$ represents equivalence in distribution, and $\left\{ \Psi  _n \right\} _{n=1}^{\infty}$ are independent and identically distributed (i.i.d.) ZUCG random variables.
\end{lemma}
\vspace{-3pt}
\begin{IEEEproof}
	Please refer to Appendix \ref{proof_lem_h_sta} for more details.
\end{IEEEproof}
Based on Lemma~\ref{lem_h_sta}, we can further deduce the statistics for the communication channel gain which is given by $g_{\mathrm{c}}\triangleq\int_{\mathcal{A} _{\mathrm{t}}}{\left| h_{\mathrm{c}}(\mathbf{t}) \right|^2}\mathrm{d}\mathbf{t}=\int_{\mathcal{Z}}{\lvert \hat{h}_{\mathrm{c}}\left( z \right) \rvert^2}\mathrm{d}z $ as follows.
\begin{corollary}\label{cor_g_sta}
The statistics of the communication channel gain satisfy
\begin{align}\label{sum_exp}
	g_{\mathrm{c}}\overset{d}{=}\sum\nolimits_{n=1}^{\infty}{\lambda _n\left| \Psi  _n \right|^2}.
\end{align}
Since $\Psi  _n\sim \mathcal{C} \mathcal{N} \left( 0,1 \right) $, we have $\left| \Psi  _n \right|^2\sim \mathsf{Exp}\left( 1 \right) $. The results of \eqref{sum_exp} thus indicate that $g_{\mathrm{c}}$ is statistically equivalent to a weighted sum of an infinite number of i.i.d. exponentially distributed random variables.
\end{corollary}
\vspace{-3pt}
\begin{IEEEproof}
Based on \eqref{hc_sta}, we have
\begin{align}
	g_{\rm{c}}&\overset{d}{=}\int_{\mathcal{Z}}{\left| \sum\nolimits_{n=1}^{\infty}{\sqrt{\lambda _n}}\phi _n\left( z \right) \Psi  _n \right|^2}\mathrm{d}z\notag\\
	&=\int_{\mathcal{Z}}{\sum\nolimits_{n=1}^{\infty}{\sqrt{\lambda _n}}\phi _n( z ) \Psi  _n\sum\nolimits_{n^{\prime}=1}^{\infty}{\sqrt{\lambda _{n^{\prime}}}\phi _{n^{\prime}}^{*}( z )}\Psi  _{n^{\prime}}^{*}}\mathrm{d}z\notag\\
	&=\sum\nolimits_{n=1}^{\infty}{\sum\nolimits_{n^{\prime}=1}^{\infty}{\sqrt{\lambda _n\lambda _{n^{\prime}}}}}\Psi  _n\Psi  _{n^{\prime}}^{*}\int_{\mathcal{Z}}{\phi _n( z ) \phi _{n^{\prime}}^{*}( z )}\mathrm{d}z\notag\\
	&=\sum\nolimits_{n=1}^{\infty}{\lambda _n\left| \Psi  _n \right|}^2.
\end{align}
\end{IEEEproof}

However, analyzing the statistics of an infinite
weighted sum of exponentially distributed random variables
is challenging. To address this, we explore more properties of the eigenvalues of $R_{\rm{c}}( z,z^{\prime} )$, i.e., $\left\{ \lambda _n \right\} _{n=1}^{\infty}$. In particular, we define
\begin{equation}
K(z,z^{\prime})\triangleq \frac{1}{2\pi}\int_{-k_0}^{k_0}{\mathrm{e}^{\mathrm{j}\kappa _z\left( z-z^{\prime} \right)}\mathrm{d}\kappa _z}=\frac{k_0}{\pi}R_{\mathrm{c}}(z,z^{\prime}),\ z,z^{\prime}\in \mathcal{Z}.
\end{equation}
Let $\varepsilon _1\ge \varepsilon _2\ge ...\ge \varepsilon _{\infty}\ge 0$ denote the eigenvalues of $K(z,z^{\prime})$, which yields $\frac{\pi}{k_0}\varepsilon _n=\lambda _n$. For an arbitrary square-integrable function $f(z')$ defined on $z'\in \mathcal{Z}$, we introduce
\begin{equation}
\hat{f}\left( z \right) \triangleq \int_{\mathcal{Z}}{K(z,z^{\prime})f\left( z^{\prime} \right) \mathrm{d}z^{\prime}}, \	z\in \mathcal{Z},
\end{equation}
which can be rewritten as follows:
\begin{equation}
\hat{f}\left( z \right) =\mathbf{1}_{\mathcal{Z}}\left( z \right) \int_{\mathcal{Z}}{g(z-z^{\prime})\mathbf{1}_{\mathcal{Z}}\left( z^{\prime} \right) f\left( z^{\prime} \right) \mathrm{d}z^{\prime}}.	
\end{equation}
Here, $\mathbf{1}_{\mathcal{Z}}\left( \cdot \right) $ represent the indicator function over the set $\mathcal{Z}$, and the function $g\left( \cdot \right) $ is defined as 
\begin{equation}
	g(x)\triangleq \frac{1}{2\pi}\int_{-k_0}^{k_0}{\mathrm{e}^{\mathrm{j}x\kappa _z}\mathrm{d}\kappa _z}, \ x\in\mathbb{R},
\end{equation}
which corresponds to the inverse Fourier transform of $\mathbf{1}_{\left[ -k_0,k_0 \right]}\left( \kappa _z \right) $. In other words, the Fourier transform of $g(x)$ is an ideal filter over the band $\left[ -k_0,k_0 \right]$. Consequently, according Landau’s eigenvalue theorem \cite{landau1980eigenvalue,landau_tit}, the eigenvalues $\left\{ \varepsilon _n \right\} _{n=1}^{\infty}$ satisfy
\begin{align}\label{epsilon}
1\ge\varepsilon _1\ge \varepsilon _2\ge ...\ge \varepsilon _{\infty}\ge 0,
\end{align}
and are governed by the effective DoFs:
\begin{align}\label{dof}
\mathsf{DoF}=\frac{1}{2\pi}\mu \left( \left[ -k_0,k_0 \right] \right) \mu \left( \mathcal{Z} \right) =\frac{2k_0L_{\rm{t}}}{2\pi}=\frac{2L_{\rm{t}}}{\lambda}.	
\end{align}
As $L_{\rm{t}}\rightarrow \infty$ or $\mathsf{DoF}\rightarrow \infty$, the eigenvalues $\left\{ \varepsilon _n \right\} _{n=1}^{\infty}$ asymptotically polarize. Specifically, for any $\varepsilon>0$, it holds that
\begin{align}\label{polarize}
\lvert n\!:\varepsilon _n\!>\varepsilon \rvert=\mathsf{DoF}+\!\left( \frac{1}{\pi ^2}\log \frac{1\!-\!\sqrt{\varepsilon}}{\sqrt{\varepsilon}} \right) \log \mathsf{DoF}+\!o( \log \mathsf{DoF} ) .
\end{align} 
\begin{remark}\label{rem_landau}
The results in \eqref{epsilon} and \eqref{polarize} indicate that as $L_{\mathrm{t}}$ increases, the leading $\mathsf{DoF}$ eigenvalues $\left\{ \varepsilon _n \right\} _{n=1}^{\mathsf{DoF}}$ approach one, while the rest are near zero, with a transition band whose width scales proportionally to $\log \mathsf{DoF}$. 
\end{remark}

These observations imply that, for small $n$, the eigenvalues $\lambda_n$ decay slowly until reaching a critical point at $n=\mathsf{DoF}$, beyond which they decrease rapidly. This step-like behavior becomes more pronounced as the physical length $L_{\rm{t}}$ increases. Since CAPAs are typically electromagnetically large arrays with $L_{\rm{t}}\gg \lambda$, the sum $\sum\nolimits_{n=1}^{\infty}{\lambda _n\left| \Psi  _n \right|^2}$ is dominated by the first $\mathsf{DoF}$ terms, which yields
\begin{align}
	g_{\mathrm{c}}\overset{d}{=}\sum\nolimits_{n=1}^{\infty}{\lambda _n\left| \Psi  _n \right|^2}\approx \sum\nolimits_{n=1}^{\mathsf{DoF}}{\lambda _n\left| \Psi  _n \right|^2}.
\end{align}
 This suggests that the communication channel gain can be approximated as a finite weighted sum of i.i.d. exponentially distributed random variables, each following $\mathsf{Exp}\left( 1 \right)$. Consequently, its PDF is given as follows.
 \vspace{-3pt}
 \begin{lemma}\label{lem_pdf}
 The PDF of communication channel gain $\int_{\mathcal{A} _{\mathrm{t}}}{\left| h_{\mathrm{c}}(\mathbf{t}) \right|^2}\mathrm{d}\mathbf{t}$ is given by \cite{sum_gamma}
 \begin{align}\label{pdf}
 	f_{g}( x ) =\frac{\lambda _{\mathsf{DoF}}^{\mathsf{DoF}}}{\prod_{n=1}^\mathsf{DoF}{\lambda _n}}\sum_{m=0}^{\infty}{\frac{\xi _mx^{\mathsf{DoF}+m-1}}{\lambda _{\mathsf{DoF}}^{\mathsf{DoF}+m}\Gamma( \mathsf{DoF}+m )}}{\rm{e}}^{-\frac{x}{\lambda _{\mathsf{DoF}}}},
 \end{align}
 where $\Gamma \left( x \right) =\int_0^{\infty}{t^{x-1}\mathrm{e}^{-t}\mathrm{d}t}$ is the Gamma function, $\xi _{0}=1$, and $\xi_m$ ($m>0$) is calculated recursively as $\xi_m=\frac{1}{m}\sum_{i=1}^{m}\xi_{m-i}{\sum_{n=1}^{\mathsf{DoF}}{(1-{\lambda_{\mathsf{DoF}}}/{\lambda_{n}})^i}}$. 
 \end{lemma}


\section{Performance Analysis}\label{section_I_CSI}
In this section, we analyze the performance of communications and sensing under the different beamforming designs.

\subsection{Sensing-Centric Design}
In the S-C design, the beamformer is designed to maximize the SR, which is given by
\begin{align}\label{SC_Beamforming_Design}
	w_{\rm{s}}({\mathbf{t}})=\argmax_{\int_{{\mathcal{A}}_{\rm{t}}}\lvert{w}({\mathbf{t}})\rvert^2{\rm{d}}{\mathbf{t}}=1}{\mathcal{R}}_{{\rm{s}}}({{w}}({\mathbf{t}}))
	=\frac{h_{\rm{s}}^{*}({\mathbf{t}})}
	{\sqrt{{\int_{{\mathcal{A}}_{\rm{t}}}\lvert h_{\rm{s}}({\mathbf{t}})\rvert^2{\rm{d}}{\mathbf{t}}}}}.
\end{align} 

\subsubsection{Performance of Sensing}
Substituting $w({\mathbf{t}})=w_{\rm{s}}({\mathbf{t}})$ into \eqref{SR_define_2_step1} yields
\begin{align}\label{ins_sc_sr_def}
	{\mathcal{R}}_{\rm{s}}^{\rm{s}}=\frac{1}{L}\log _2\left( 1+\overline{\gamma }_{\mathrm{s}}L\alpha _{\mathrm{s}}\int_{\mathcal{A} _{\mathrm{r}}}{\left| h_{\mathrm{s}}(\mathbf{r}) \right|^2}\mathrm{d}\mathbf{r}\int_{\mathcal{A} _{\mathrm{t}}}{\left| h_{\mathrm{s}}(\mathbf{t}) \right|^2}\mathrm{d}\mathbf{t} \right) .
\end{align}
A closed-form expression for ${\mathcal{R}}_{\rm{s}}^{\rm{s}}$ is given as follows.
\vspace{-5pt}
\begin{theorem}\label{the_ins_sc_sr}
	In the S-C design, the SR can be expressed as follows:
	\begin{equation}\label{ins_sc_sr}
			{\mathcal{R} _{{\mathrm{s}}}^{\mathrm{s}}}=\frac{1}{L}\log _2\left( 1+{\overline{\gamma }_{\mathrm{s}}L\alpha _{\mathrm{s}}G_{\mathrm{r}}G_{\mathrm{t}}}\right).
	\end{equation}
	Here,
	\begin{equation}
	G_{\mathrm{t}}=\frac{\zeta \left( \frac{d}{2},\frac{d}{2}+L_{\mathrm{t}} \right)}{4\pi \sqrt{p_{\mathrm{s},x}^{2}+p_{\mathrm{s},y}^{2}}}, \
	G_{\mathrm{r}}=\frac{\zeta \left( -\frac{d}{2}-L_{\mathrm{r}},-\frac{d}{2} \right)}{4\pi \sqrt{p_{\mathrm{s},x}^{2}+p_{\mathrm{s},y}^{2}}}, 
	\end{equation}
	 where $\zeta \left( x,y \right) \triangleq \arctan \left( \frac{y-p_{\mathrm{s},z}}{\sqrt{p_{\mathrm{s},x}^{2}+p_{\mathrm{s},y}^{2}}} \right) -\arctan \left( \frac{x-p_{\mathrm{s},z}}{\sqrt{p_{\mathrm{s},x}^{2}+p_{\mathrm{s},y}^{2}}} \right) $.
\end{theorem}
\vspace{-5pt}
\begin{IEEEproof}
	Please refer to Appendix~\ref{proof_the_ins_sc_sr} for more details.
\end{IEEEproof}
To gain further insights, we then examine the asymptotic SR in the high-SNR regime to obtain the maximal multiplexing gain, i.e., the high-SNR slope, which is defined as 
\begin{equation}
	\mathcal{S} _i=\lim_{P\rightarrow \infty} \frac{\mathcal{R} _i}{\log _2P}, \ i\in\{\mathrm{c},\mathrm{s}\}.
\end{equation}

\begin{corollary}
As $P\rightarrow \infty$, the asymptotic SR of the S-C design satisfies
	\begin{equation}
		\mathcal{R} _{\mathrm{s}}^{\mathrm{s}}\simeq \frac{1}{L}\log _2\overline{\gamma }_{\mathrm{s}}+\frac{1}{L}\log _2\left( L\alpha _{\mathrm{s}}G_{\mathrm{r}}G_{\mathrm{t}} \right). 
	\end{equation}
	This suggests that the high-SNR slope of the SR achieved by the S-C design is $\frac{1}{L}$.
\end{corollary}
\vspace{-5pt}
\begin{IEEEproof}
	The results can be easily obtained based on the fact $\lim_{x\rightarrow \infty} \log _2\left( 1+x \right) \simeq \log _2x$.
\end{IEEEproof}

\subsubsection{Performance of Communications}
By inserting $w(\mathbf{t})=w_\mathrm{s}(\mathbf{t})$ into \eqref{CR_define}, the CR under the S-C design is expressed as
\begin{align}\label{def_sc_ecr}
	{\mathcal{R} _{{\mathrm{c}}}^{\mathrm{s}}}=\log _2\left( 1+\overline{\gamma }_{\mathrm{c}}\frac{\left|\rho  \right|^2}{\int_{\mathcal{A} _{\mathrm{t}}}{\left| h_{\mathrm{s}}(\mathbf{t}) \right|^2}\mathrm{d}\mathbf{t}} \right) ,
\end{align}	
where $\rho\triangleq\int_{\mathcal{A} _{\mathrm{t}}}{h_{\mathrm{c}}}(\mathbf{t})h_{\mathrm{s}}^{*}(\mathbf{t})\mathrm{d}\mathbf{t}$ is a random variable.
In the following theorem, we derive the ECR.
\vspace{-3pt}
\begin{theorem}\label{the_ins_sc_cr}
In the S-C design, the ECR, i.e.,  $\overline{\mathcal{R} }_{\mathrm{c}}^{\mathrm{s}}=\mathbb{E}\{{\mathcal{R} }_{\mathrm{c}}^{\mathrm{s}}\}$, can be expressed as follows:	
\begin{equation}\label{ins_sc_cr}
	\overline{\mathcal{R} }_{\mathrm{c}}^{\mathrm{s}}=-\frac{1}{\ln 2}\mathrm{e}^{\frac{G_{\mathrm{t}} }{\overline{\gamma }_{\mathrm{c}}\Xi}}\mathrm{Ei}\left( -\frac{G_\mathrm{t} }{\overline{\gamma }_{\mathrm{c}}\Xi} \right) ,
\end{equation}
where $\Xi  ={\int_{\mathcal{Z}}{\int_{\mathcal{Z}}{\hat{h}_{\mathrm{s}}\left( z \right) R_{\mathrm{c}}\left( z,z^{\prime} \right) \hat{h}_{\mathrm{s}}^{*}\left( z^{\prime} \right)}}\mathrm{d}z\mathrm{d}z^{\prime}}$ with $\hat{h}_{\mathrm{s}}\left( z \right)$ defined as \eqref{hshat}, and $\mathrm{Ei}\left( x \right) =-\int_{-x}^{\infty}{\frac{{\rm{e}}^{-t}}{t}dt}$ is the exponential integral function.
\end{theorem}
\vspace{-5pt}
\begin{IEEEproof}
	Please refer to Appendix~\ref{proof_the_ins_sc_cr} for more details.
\end{IEEEproof}

\begin{corollary}\label{cor_ins_sc_cr}
	The high-SNR ECR under the S-C design satisfies
	\begin{align}\label{asy_sc_cr}
		\overline{\mathcal{R} }_{\mathrm{c}}^{\mathrm{s}}\simeq \log _2\overline{\gamma }_{\mathrm{c}}-\log _2\frac{G_{\mathrm{t}}}{\Xi}   -\frac{\mathcal{C}}{\ln 2},
	\end{align}	
	where $\mathcal{C}\approx0.5772$ is the Euler constant. This suggests that the high-SNR slope of the CR achieved by the S-C design is one.
\end{corollary}
\vspace{-5pt}
\begin{IEEEproof}
	When $P\rightarrow \infty$, the CR can approximated as ${\mathcal{R} _{\mathrm{c}}^{\mathrm{s}}}\simeq \log _2\left( \overline{\gamma }_{\mathrm{c}}\frac{\left| \int_{\mathcal{A} _{\mathrm{t}}}{h_{\mathrm{c}}}(\mathbf{t})h_{\mathrm{s}}^{*}(\mathbf{t})\mathrm{d}\mathbf{t} \right|^2}{\int_{\mathcal{A} _{\mathrm{t}}}{\left| h_{\mathrm{s}}(\mathbf{t}) \right|^2}\mathrm{d}\mathbf{t}} \right) $, and thus the ECR can be calculated based on Appendix~\ref{proof_the_ins_sc_cr} and \cite[(4.331.1)]{integral}.
\end{IEEEproof}

Next, we analyze the OP, which is defined as
\begin{equation}
	\mathcal{P}=\mathrm{Pr}\left({\mathcal{R} _{\mathrm{c}}}<\mathcal{R} _{0}\right),
\end{equation}
where $\mathcal{R} _{0}$ denotes the target rate.

\begin{theorem}\label{the_inc_sc_op}
The communication OP in the S-C design can be expressed in closed form as follows:
\begin{equation}\label{ins_sc_op}
\mathcal{P} _{\mathrm{s}}=1-\mathrm{e}^{-\frac{G_{\mathrm{t}} }{\overline{\gamma }_{\mathrm{c}}\Xi}\left( 2^{\mathcal{R} _0}-1 \right)}.	
\end{equation} 	
\end{theorem}
\vspace{-3pt}
\begin{IEEEproof}
	Please refer to Appendix~\ref{proof_the_ins_sc_cr} for more details.
\end{IEEEproof}

We then consider the high-SNR OP to obtain the maximal diversity gain, i.e., the diversity order, which is defined as
\begin{equation}
	\mathcal{D} =-\lim_{P\rightarrow \infty} \frac{\log _2\mathcal{P}}{\log _2P}.	
\end{equation}

\begin{corollary}
	As $P\rightarrow \infty$, the asymptotic OP of the S-C design satisfies
	\begin{equation}\label{ins_sc_op_asy}
		\mathcal{P} _{\mathrm{s}}\simeq \frac{G_{\mathrm{t}} }{\overline{\gamma }_{\mathrm{c}}\Xi}\left( 2^{\mathcal{R} _0}-1 \right) .
	\end{equation}
	This suggests that the diversity order achieved by the S-C design is one.
\end{corollary}
\vspace{-3pt}
\begin{IEEEproof}
\eqref{ins_sc_op_asy} is obtained by applying the fact $\lim_{x\rightarrow 0}\mathrm{e}^{-x}\simeq 1-x$ to \eqref{ins_sc_op}.
\end{IEEEproof}

\subsection{Communications-Centric Design}
In the C-C design, the continuous beamformer is designed to maximize the CR, given by
\begin{align}\label{CC_Beamforming_Design}
	w_{\rm{c}}({\mathbf{t}})=\argmax_{\int_{{\mathcal{A}}_{\rm{t}}}\lvert{w}({\mathbf{t}})\rvert^2{\rm{d}}{\mathbf{t}}=1}{\mathcal{R}}_{{\rm{c}}}({{w}}({\mathbf{t}}))=\frac{h_{\rm{c}}^{*}({\mathbf{t}})}
	{\sqrt{{\int_{{\mathcal{A}}_{\rm{t}}}\lvert h_{\rm{c}}({\mathbf{t}})\rvert^2{\rm{d}}{\mathbf{t}}}}}.
\end{align} 
\subsubsection{Performance of Communications}
Substituting $w({\mathbf{t}})=w_{\rm{c}}({\mathbf{t}})$ into \eqref{CR_define}, the CR is expressed as follows:
\begin{align}
{\mathcal{R} _{\mathrm{c}}^{\mathrm{c}}}=\log _2\left( 1+\overline{\gamma}_{\mathrm{c}}g_{\mathrm{c}} \right) .
\end{align}
Based on \textbf{Lemma~\ref{lem_pdf}}, the following theorem can be found. 
\vspace{-5pt}
\begin{theorem}\label{the_ins_cc_cr}
In the C-C design, the ECR $\mathbb{E}\{{\mathcal{R} }_{\mathrm{c}}^{\mathrm{c}}\}$ is given by
\begin{equation}\label{ins_cc_cr}
	\begin{split}
		\overline{\mathcal{R} }_{\mathrm{c}}^{\mathrm{c}}&=\frac{\lambda _{\mathsf{DoF}}^{\mathsf{DoF}}}{\ln2\prod_{n=1}^{\mathsf{DoF}}{\lambda _n}}\sum_{m=0}^{\infty}{\sum_{u=0}^{\mathsf{DoF}+m-1}\!{\frac{\xi _m}{(\mathsf{DoF}+m\!-\!u\!-\!1)!}}}\\
		&\times\left[ \frac{(-1)^{\mathsf{DoF}+m-u}\mathrm{e}^{1/\left( \overline{\gamma }_{\mathrm{c}}\lambda _{\mathsf{DoF}} \right)}}{(\overline{\gamma }_{\mathrm{c}}\lambda _{\mathsf{DoF}})^{\mathsf{DoF}+m-u-1}}\mathrm{Ei}\left( -\frac{1}{\overline{\gamma }_{\mathrm{c}}\lambda _{\mathsf{DoF}}} \right)\right.\\ &\left.+\!\sum_{l=1}^{\mathsf{DoF}+m-u-1}\!\!{(l-1)!\left( -\frac{1}{\overline{\gamma }_{\mathrm{c}}\lambda _{\mathsf{DoF}}} \right) ^{\mathsf{DoF}+m-u-1-l}} \right] .
	\end{split}	
\end{equation}
\end{theorem}
\vspace{-2pt}
\begin{IEEEproof}
The ECR is calculated as 
\begin{align}\label{cc_ecr}
	\overline{\mathcal{R} }_{\mathrm{c}}^{\mathrm{c}}&=\mathbb{E} \left\{ \log _2\left( 1+\overline{\gamma }_{\mathrm{c}}g_{\mathrm{c}} \right) \right\}\notag\\
	& =\frac{1}{\ln 2}\int_0^{\infty}{\ln\left( 1+\overline{\gamma }_{\mathrm{c}}x \right) f_g\left( x \right) \mathrm{d}x}.
\end{align}
By substituting \eqref{pdf} into \eqref{cc_ecr} and using \cite[(4.337.5)]{integral}, the results of \eqref{ins_cc_cr} can be obtained.
\end{IEEEproof}

\begin{corollary}\label{cor_ins_cc_cr}
The asymptotic ECR achieve by the C-C design in the high-SNR regime is given by
\begin{align}
\overline{\mathcal{R} }_{\mathrm{c}}^{\mathrm{c}}\simeq \log _2\!\overline{\gamma }_{\mathrm{c}}+\frac{\lambda _{\mathsf{DoF}}^{\mathsf{DoF}}\log _2\!\mathrm{e}}{\prod_{n=1}^{\mathsf{DoF}}{\lambda _n}}\!\sum_{m=0}^{\infty}{\xi _m\left( \psi (\mathsf{DoF}\!+\!m)+\ln(\lambda _{\mathsf{DoF}}) \right)},
\end{align}	
where $\psi \left( x \right) =\frac{\mathrm{d}}{\mathrm{d}x}\ln \Gamma \left( x \right) $ is the digamma function. This suggests that the high-SNR slope of the CR achieved by the C-C design is one.
\end{corollary}
\vspace{-5pt}
\begin{IEEEproof}
When $P\rightarrow \infty$, the CR can approximated as ${\mathcal{R} _{\mathrm{c}}^{\mathrm{c}}}\approx\log _2\left( \overline{\gamma}_{\mathrm{c}}g_{\mathrm{c}} \right)$, and thus the ECR is given by
\begin{equation}
\overline{\mathcal{R} }_{\mathrm{c}}^{\mathrm{c}}\simeq\log _2\overline{\gamma }_{\mathrm{c}}+\int_0^{\infty}{\log _2x  f_g\left( x \right) \mathrm{d}x},
\end{equation}
which can be calculated by using \cite[(4.352.1)]{integral}.	
\end{IEEEproof}

We then study the OP. The following theorem provides a closed-form expression of the OP.
\begin{theorem}
The OP of the C-C design can be expressed as
\begin{align}\label{ins_cc_op}
\mathcal{P} _{\mathrm{c}}=\frac{\lambda _{\mathsf{DoF}}^{\mathsf{DoF}}}{\prod_{n=1}^{\mathsf{DoF}}{\lambda _n}}\sum_{m=0}^{\infty}{\frac{\xi _m\gamma (\mathsf{DoF}+m,\frac{2^{\mathcal{R} _0}-1}{\overline{\gamma }_{\mathrm{c}}\lambda _{\mathsf{DoF}}})}{\left( \mathsf{DoF}+m-1 \right) !}},
\end{align}
where $\gamma( s,x ) =\int_0^x{t^{s-1}{\rm{e}}^{-t}\mathrm{d}t}$ represents the lower incomplete gamma function.
\end{theorem}
\vspace{-5pt}
\begin{IEEEproof}
By performing some manipulations, the OP can be written as $\mathcal{P} _{\mathrm{c}}=F_g\left( \frac{2^{\mathcal{R} _0}-1}{\overline{\gamma }_{\mathrm{c}}} \right)$, where $F_g\left( \cdot \right) $ denotes the CDF of $g_{\mathrm{c}}$. With the aid of \cite[(3.351.1)]{integral}, $F_g\left( x \right) $ can be calculated as follows:
\begin{align}
	F_g( x ) =\!\int_0^x\!{f_g\left( t \right) \mathrm{d}}t=\frac{\lambda _{\mathsf{DoF}}^{\mathsf{DoF}}}{\prod_{n=1}^{\mathsf{DoF}}{\lambda _n}}\sum_{m=0}^{\infty}{\frac{\xi _m\gamma (\mathsf{DoF}+m,\frac{x}{\lambda _{\mathsf{DoF}}})}{\left( \mathsf{DoF}+m-1 \right) !}},
\end{align}
which yields the results of \eqref{ins_cc_op}.
\end{IEEEproof}

\begin{corollary}
As $P\rightarrow \infty$, the asymptotic OP of the C-C design satisfies
\begin{equation}\label{ins_cc_op_asy}
	\mathcal{P} _{\mathrm{c}}\simeq \frac{\left( 2^{\mathcal{R} _0}-1 \right) ^{\mathsf{DoF}}}{\overline{\gamma }_{\mathrm{c}}^{\mathsf{DoF}} \mathsf{DoF} !\prod_{n=1}^{\mathsf{DoF}}{\lambda _n}}.
\end{equation}
This suggests that the diversity order achieved by the C-C design is $\mathsf{DoF}$.
\end{corollary}
\vspace{-5pt}
\begin{IEEEproof}
	The results of \eqref{ins_cc_op_asy} is obtained by applying the fact $\lim_{x\rightarrow 0} \gamma \left( s,x \right) \simeq \frac{x^s}{s}$ based on \cite[(8.354.1)]{integral} to \eqref{ins_cc_op}.
\end{IEEEproof}


\subsubsection{Performance of Sensing}
When the C-C beamformer $w_{\mathrm{c}}(\mathbf{t})$ is employed, the SR is given by
\begin{equation}\label{def_ins_cc_sr}
		{\mathcal{R} _{{\mathrm{s}}}^{\mathrm{c}}}\!=\!\frac{1}{L}\log _2\!\left( \!1\!+\!L\alpha _{\mathrm{s}}\overline{\gamma }_{\mathrm{s}}\int_{\mathcal{A} _{\mathrm{r}}}\!\!{\left| h_{\mathrm{s}}(\mathbf{r}) \right|^2}\mathrm{d}\mathbf{r}\frac{\left| \int_{\mathcal{A} _{\mathrm{t}}}\!{h_{\mathrm{s}}}(\mathbf{t})h_{\mathrm{c}}^{*}(\mathbf{t})\mathrm{d}\mathbf{t} \right|^2}{\int_{{\mathcal{A}}_{\rm{t}}}\lvert h_{\rm{c}}({\mathbf{t}})\rvert^2{\rm{d}}{\mathbf{t}}} \right),
\end{equation}
which can be further rewritten as follows.
\vspace{-3pt}
\begin{lemma}\label{lem_ins_cc_sr}
The SR under the C-C design can be written as
\begin{equation}
	\mathcal{R} _{\mathrm{s}}^{\mathrm{c}}=\frac{1}{L}\log _2\left( \sum\nolimits_{n=1}^{\mathsf{DoF}_Q}{\nu _n\left| \Phi _n \right|^2} \right) -\frac{1}{L}\log _2g_{\mathrm{c}},
\end{equation} 
where $\left\{ \nu  _n \right\} _{n=1}^{\infty}$ are the eigenvalues of $Q\left( z_1,z_2 \right) \triangleq \int_{\mathcal{Z}}\int_{\mathcal{Z}}\tilde{R}_{\mathrm{c}}\left( z_1,z \right) ( \delta \left( z-z^{\prime} \right) +L\alpha _{\mathrm{s}}\overline{\gamma }_{\mathrm{s}}G_\mathrm{r}\hat{h}_{\mathrm{s}}\left( z \right) \hat{h}_{\mathrm{s}}^{*}\left( z^{\prime} \right) ) \tilde{R}_{\mathrm{c}}\left( z^{\prime},z_2 \right)\mathrm{d}z\mathrm{d}z^{\prime}$ with $\tilde{R}_{\mathrm{c}}\left( z,z_1 \right) =\sum\nolimits_{n=1}^{\mathsf{DoF}}{\sqrt{\lambda _n}}\phi _n\left( z \right) \phi _{n}^{*}\left( z_1 \right) $, $\mathsf{DoF}_Q$ denotes the DoF of $Q\left( z_1,z_2 \right) $, and $\left\{ \Phi  _n \right\} _{n=1}^{\infty}$ are i.i.d. ZUCG random variables. Note that $\left\{ \nu  _n \right\} _{n=1}^{\infty}$ can be obtained using a method similar to that described in Appendix~\ref{eigenvalue}.   
\end{lemma}
\vspace{-3pt}
\begin{IEEEproof}
	Please refer to Appendix~\ref{proof_lem_ins_cc_sr} for more details.
\end{IEEEproof}
Based on Lemma~\ref{lem_ins_cc_sr}, a closed-form expression for the average SR defined as $\overline{\mathcal{R}}_{\mathrm{s}}^{\mathrm{c}}=\mathbb{E} \left\{ \mathcal{R} _{\mathrm{s}}^{\mathrm{c}} \right\} $ is derived as follows.
\vspace{-3pt}
\begin{theorem}
In the C-C design, the average SR is given by
\begin{align}\label{ins_cc_sr}
	\overline{\mathcal{R} }_{\mathrm{s}}^{\mathrm{c}}&=\frac{\nu _{\mathsf{DoF}_Q}^{\mathsf{DoF}_Q}\log _2\mathrm{e}}{L\prod_{n=1}^{\mathsf{DoF}_Q}{\nu _n}}\sum_{m=0}^{\infty}{\xi _m\left( \psi (\mathsf{DoF}_Q\!+m)+\ln(\nu _{\mathsf{DoF}_Q}) \right)}\notag\\
	&-\frac{\lambda _{\mathsf{DoF}}^{\mathsf{DoF}}\log _2\mathrm{e}}{L\prod_{n=1}^{\mathsf{DoF}}{\lambda _n}}\sum_{m=0}^{\infty}{\xi _m\left( \psi (\mathsf{DoF}\!+m)+\ln(\lambda _{\mathsf{DoF}}) \right)}.
\end{align}	
\end{theorem}
\begin{IEEEproof}
The results are obtained by using \cite[(4.352.1)]{integral}.	
\end{IEEEproof}

\vspace{-5pt}
\begin{corollary}\label{cor_ins_cc_sr}
	The high-SNR average SR under the C-C design satisfies
	\begin{equation}
		\begin{split}
			\overline{\mathcal{R}}_{\mathrm{s}}^{\mathrm{c}}&\simeq \frac{1}{L}\left( \log _2\overline{\gamma }_{\mathrm{c}}+\log _2({L\alpha _{\mathrm{s}}G_{\mathrm{r}}\Xi}  )-\frac{\mathcal{C}}{\ln 2} \right) \\
			&-\frac{\lambda _{\mathsf{DoF}}^{\mathsf{DoF}}\log _2\mathrm{e}}{L\prod_{n=1}^{\mathsf{DoF}}{\lambda _n}}\sum_{m=0}^{\infty}{\xi _m\left( \psi (\mathsf{DoF}\!+m)+\ln(\lambda _{\mathsf{DoF}}) \right)}.
		\end{split}		
	\end{equation}	
	This suggests that the high-SNR slope of the SR achieved by the C-C design is $\frac{1}{L}$.
\end{corollary}
\vspace{-5pt}
\begin{IEEEproof}
	When $P\rightarrow \infty$, \eqref{def_ins_cc_sr} can be rewritten as follows:
	\begin{equation}
	\mathcal{R} _{\mathrm{s}}^{\mathrm{c}}\simeq \frac{1}{L}\log _2\left( L\alpha _{\mathrm{s}}\overline{\gamma }_{\mathrm{s}}G_\mathrm{r}\left| \rho \right|^2 \right) -\frac{1}{L}\log _2g_{\mathrm{c}} .
	\end{equation}
	The average SR can be then calculated by following the steps similar to the proofs of Corollary \ref{cor_ins_sc_cr} and \ref{cor_ins_cc_cr}.	
\end{IEEEproof}

By comparing the sensing and communication performance of the S-C and C-C designs, we have the following observations.
\begin{remark}\label{compare}
The S-C and C-C designs achieve the same multiplexing gains for both sensing and communications, whereas the C-C design provides a higher diversity gain than the S-C design.
\end{remark}

\subsection{Pareto-Optimal Design}
In real-world implementations, the beamformer $w({\mathbf{t}})$ can be tailored to satisfy different quality-of-service demands, inevitably leading to a balance between communication and sensing performance. To assess this balance, we analyze the Pareto boundary of the achievable SR–CR region in the CAPA-based ISAC framework. This boundary consists of SR–CR pairs where enhancing one rate necessarily results in the degradation of the other. It effectively depicts the optimal trade-off between sensing and communication functionalities and offers key insights for system optimization.

It should be emphasized that the SR monotonically increases with respect to $\lvert \int_{\mathcal{A} _{\mathrm{t}}}{h_{\mathrm{s}}}(\mathbf{t})w(\mathbf{t})\mathrm{d}\mathbf{t}  \rvert^2\triangleq{\Upsilon}_{\rm{s}}({{w}}({\mathbf{t}}))$, while the CR exhibits a similar monotonic relationship with $\lvert\int_{{\mathcal{A}}_{\rm{t}}}h_{\rm{c}}(\mathbf{t})w({\mathbf{t}}){\rm{d}}{\mathbf{t}}\rvert^2\triangleq{\Upsilon}_{\rm{c}}({{w}}({\mathbf{t}}))$. Consequently, the achievable SR–CR region can be characterized through the $\hat{\gamma}_{\rm{c}}({{w}}({\mathbf{t}}))$-$\hat{\gamma}_{\rm{s}}({{w}}({\mathbf{t}}))$ space. Specifically, any point lying on the Pareto boundary of this region can be determined by solving the following optimization problem:
\begin{subequations}\label{pareto}
\begin{align}
\max_{w({\mathbf{t}}),\Upsilon}~~&\Upsilon\\
{\rm{s.t.}}~~&{\Upsilon}_{\rm{s}}({{w}}({\mathbf{t}}))\geq\tau\Upsilon,~\Upsilon_{\rm{c}}({{w}}({\mathbf{t}}))\geq(1-\tau)\Upsilon,\\
&\int_{{\mathcal{A}}_{\rm{t}}}\lvert{w}({\mathbf{t}})\rvert^2{\rm{d}}{\mathbf{t}}=1,
\end{align}
\end{subequations}
where $\tau\in[0,1]$ is a trade-off parameter. The entire Pareto boundary is determined by solving the above problem with $\tau$ varying from $0$ to $1$, which belongs to the class of non-convex integral-based functional programming. To solve such a challenging problem, we introduce a \textit{subspace approach} as follows.
\vspace{-5pt}
\begin{lemma}\label{Lemma_Subspace}
For a given $\tau$, the Pareto-optimal beamformer lies in the subspace spanned by $\{h_{\rm{s}}^{*}({\mathbf{t}}),h_{\rm{c}}^{*}(\mathbf{t})\}$, i.e., $w({\mathbf{t}})=ah_{\rm{c}}^{*}({\mathbf{t}})+bh_{\rm{s}}^{*}(\mathbf{t})$, where $a,b\in\mathbb{C}$.
\end{lemma}
\vspace{-5pt}
\begin{IEEEproof}
Please refer to Appendix \ref{proof_Lemma_Subspace} for more details.
\end{IEEEproof}
Based on Lemma \ref{Lemma_Subspace}, we denote $\{\phi_1(\mathbf{t}),\phi_2(\mathbf{t})\}$ as an orthonormal basis for the signal subspace spanned by $\{h_{\rm{s}}^{*}({\mathbf{t}}),h_{\rm{c}}^{*}(\mathbf{t})\}$, which satisfies
\begin{align}
\int_{{\mathcal{A}}_{\rm{t}}}\phi_{i}^{*}(\mathbf{t})\phi_{j}(\mathbf{t}){\rm{d}}{\mathbf{t}}=\delta_{i,j},\quad\forall i,j\in\{1,2\},
\end{align}
where $\delta_{i,j}$ denotes the Kronecker delta. The beamformer and channels can be expressed in terms of the basis as follows:
\begin{subequations}\label{Subspace_Approach_Transform_Basic_Variables}
\begin{align}
w({\mathbf{t}})&=w_1\phi_1(\mathbf{t})+w_2\phi_2(\mathbf{t}),\\
h_{\rm{s}}^{*}({\mathbf{t}})&=h_{\mathrm{s},1}\phi_1(\mathbf{t})+h_{\mathrm{s},2}\phi_2(\mathbf{t}),\\
h_{\rm{c}}^{*}(\mathbf{t})&=h_{\mathrm{c},1}\phi_1(\mathbf{t})+h_{\mathrm{c},2}\phi_2(\mathbf{t}).
\end{align}
\end{subequations}
Defining the vector representations: ${\mathbf{w}}=[w_1,w_2]^{\mathsf{T}}$, ${\mathbf{h}}_{\mathrm{s}}=[h_{\mathrm{s},1},h_{\mathrm{s},2}]^{\mathsf{T}}$, and ${\mathbf{h}}_{\mathrm{c}}=[h_{\mathrm{c},1},h_{\mathrm{c},2}]^{\mathsf{T}}$, we have $\left\| \mathbf{h}_{\mathrm{s}} \right\| ^2=\int_{\mathcal{A} _{\mathrm{t}}}{\!}\left| h_{\mathrm{s}}(\mathbf{t}) \right|^2\mathrm{d}\mathbf{t}=G_{\mathrm{t}}$, $\left\| \mathbf{h}_{\mathrm{c}} \right\| ^2=\int_{\mathcal{A} _{\mathrm{t}}}{\!}\left| h_{\mathrm{c}}(\mathbf{t}) \right|^2\mathrm{d}\mathbf{t}=g_{\mathrm{c}}$, and $\mathbf{h}_{\mathrm{s}}^{\mathsf{H}}\mathbf{h}_{\mathrm{c}}=\int_{\mathcal{A} _{\mathrm{t}}}{h_{\mathrm{c}}(\mathbf{t})h_{\mathrm{s}}^{*}(\mathbf{t})}\mathrm{d}\mathbf{t}=\rho$. Furthermore, the integrals can be transformed into the following forms:
\begin{subequations}
\begin{align}
	&\int_{{\mathcal{A}}_{\rm{t}}}h_{\rm{s}}({\mathbf{t}})
	{{w}}({\mathbf{t}}){\rm{d}}{\mathbf{t}}={\mathbf{h}}_{\mathrm{s}}^{\mathsf{T}}{\mathbf{w}}, \
	\int_{{\mathcal{A}}_{\rm{t}}}h_{\rm{c}}(\mathbf{t})w({\mathbf{t}}){\rm{d}}{\mathbf{t}}
	={\mathbf{h}}_{\mathrm{c}}^{\mathsf{T}}{\mathbf{w}},\\
	&\int_{{\mathcal{A}}_{\rm{t}}}\lvert{w}({\mathbf{t}})\rvert^2{\rm{d}}{\mathbf{t}}=\lVert{\mathbf{w}}\rVert^2=1,
\end{align}
\end{subequations}
Consequently, problem \eqref{pareto} can be equivalently rewritten as follows:
\begin{align}\label{pareto_problme_2}
\max_{{\mathbf{w}},\Upsilon}\Upsilon~~{\rm{s.t.}}~\lvert{\mathbf{h}}_{\mathrm{s}}^{\mathsf{T}}{\mathbf{w}}\rvert^2\geq\tau\Upsilon,~
\lvert{\mathbf{h}}_{\mathrm{c}}^{\mathsf{T}}{\mathbf{w}}\rvert^2\geq(1\!-\!\tau)\Upsilon,~
\lVert{\mathbf{w}}\rVert^2\!=\!1.
\end{align}
By employing the subspace approach, the originally intractable functional programming problem is converted to problem \eqref{pareto_problme_2}, which is a classical vector-based optimization problem that can be solved using the Karush–Kuhn–Tucker (KKT) conditions.
\vspace{-5pt}
\begin{lemma}\label{lem_w_star}
For a given $\tau$, the Pareto-optimal solution for problem \eqref{pareto_problme_2} is given by
\begin{align}\label{w_star}
\mathbf{w}_{\tau}=\begin{cases}
	\frac{\mathbf{h}_{\mathrm{c}}^*}{\sqrt{g_{\mathrm{c}}}}&		\tau\in \left[ 0,\frac{\left| \rho  \right|^2}{\left| \rho  \right|^2+g_{\mathrm{c}}^{2}} \right]\\
	\frac{\epsilon _1}{\varsigma }\mathbf{h}_{\mathrm{s}}^*+\frac{\epsilon _2\mathrm{e}^{-\mathrm{j}\angle \rho }}{\varsigma }\mathbf{h}_{\mathrm{c}}^*&		\tau \in \left( \frac{\left| \rho  \right|^2}{\left| \rho  \right|^2+g_{\mathrm{c}}^{2}},\frac{G_{\mathrm{t}}^{2}}{G_{\mathrm{t}}^{2}+\left| \rho  \right|^2} \right)\\
	\frac{\mathbf{h}_{\mathrm{s}}^*}{\sqrt{G_{\mathrm{t}}}}&		\tau \in \left[ \frac{G_{\mathrm{t}}^{2}}{G_{\mathrm{t}}^{2}+\left| \rho  \right|^2},1 \right]\\
\end{cases},    
\end{align}
where $\epsilon_1=\frac{\sqrt{\tau}g_{\mathrm{c}}-\sqrt{\left( 1-\tau \right)}\left| \rho  \right|}{\left( 1-\tau \right) G_{\mathrm{t}}+\tau g_{\mathrm{c}}-2\sqrt{\tau \left( 1-\tau \right)}\left| \rho  \right|}$, $\epsilon _2=\frac{\sqrt{\left( 1-\tau \right)}G_{\mathrm{t}}-\sqrt{\tau}\left| \rho  \right|}{\left( 1-\tau \right) G_{\mathrm{t}}+\tau g_{\mathrm{c}}-2\sqrt{\tau \left( 1-\tau \right)}\left| \rho  \right|}$, and $\varsigma  =\sqrt{\epsilon _{1}^{2}G_{\mathrm{t}}+\epsilon _{2}^{2}g_{\mathrm{c}}+2\epsilon _1\epsilon _2\left| \rho  \right|}$ is for normalization.
\end{lemma}
\vspace{-5pt}
\begin{IEEEproof}
Please refer to Appendix \ref{proof_lem_w_star} for more details.
\end{IEEEproof}
Having obtain $\mathbf{w}_{\tau}$, based on \eqref{Subspace_Approach_Transform_Basic_Variables}, we can derive the Pareto-optimal continuous beamformer for problem \eqref{pareto}.
\vspace{-5pt}
\begin{theorem}\label{the_ins_pareto}
For a given $\tau$, the Pareto-optimal continuous beamformer is given by
\begin{equation}\label{Pareto_Optimal_Continuous_Beamformer}
\resizebox{1\hsize}{!}{$w_{\tau}( \mathbf{t} ) \!=\begin{cases}
	w_{\rm{c}}({\mathbf{t}})&		\tau \in \left[ 0,\frac{\left| \rho  \right|^2}{\left| \rho  \right|^2+g_{\mathrm{c}}^{2}} \right]\\
	\frac{\epsilon _1}{\varsigma }h_{\mathrm{s}}^{*}( \mathbf{t} ) \!+\!\frac{\epsilon _2\mathrm{e}^{-\mathrm{j}\angle \rho }}{\varsigma }h_{\rm{c}}^{*}( \mathbf{t} )&		\tau \in \!\left( \frac{\left| \rho  \right|^2}{\left| \rho  \right|^2+g_{\mathrm{c}}^{2}},\frac{G_{\mathrm{t}}^{2}}{G_{\mathrm{t}}^{2}+\left| \rho  \right|^2} \right)\\
	w_{\rm{s}}({\mathbf{t}})&		\tau \in \left[ \frac{G_{\mathrm{t}}^{2}}{G_{\mathrm{t}}^{2}+\left| \rho  \right|^2},1 \right]\\
\end{cases}.$}
\end{equation}
\end{theorem}
\vspace{-5pt}
\begin{remark}
The results of Theorem \ref{the_ins_pareto} suggest that the Pareto-optimal beamformer is a weighted combination of $\{h_{\rm{s}}^{*}({\mathbf{t}}),h_{\rm{c}}^{*}(\mathbf{t})\}$. In particular, it reduces to the S-C beamformer when $\tau \in \left[ \frac{G_{\mathrm{t}}^{2}}{G_{\mathrm{t}}^{2}+\left| \rho  \right|^2},1 \right]$, and to the C-C beamformer when $\tau \in \left[ 0,\frac{\left| \rho  \right|^2}{\left| \rho  \right|^2+g_{\mathrm{c}}^{2}} \right]$.
\end{remark}
\vspace{-5pt}
For a given $\tau$, let $(\mathcal{R}_{\rm{s}}^\tau,\overline{\mathcal{R}}_{\rm{c}}^\tau)$ represent the SR-CR pair located on the Pareto boundary. Therefore, we have $\mathcal{R}_{\rm{s}}^\tau\in[\mathcal{R}_{\rm{s}}^{\rm{c}},\mathcal{R}_{\rm{s}}^{\rm{s}}]$ and $\overline{\mathcal{R}}_{\rm{c}}^\tau\in[\overline{\mathcal{R}}_{\rm{c}}^{\rm{s}},\overline{\mathcal{R}}_{\rm{c}}^{\rm{c}}]$, where $\mathcal{R}_{\rm{s}}^0=\mathcal{R}_{\rm{s}}^{\rm{c}}$, $\mathcal{R}_{\rm{s}}^1=\mathcal{R}_{\rm{s}}^{\rm{s}}$, $\overline{\mathcal{R}}_{\rm{c}}^0=\overline{\mathcal{R}}_{\rm{c}}^{\rm{c}}$, and $\overline{\mathcal{R}}_{\rm{c}}^1=\overline{\mathcal{R}}_{\rm{c}}^{\rm{s}}$. Based on the Sandwich theorem, we can obtain the high-SNR slopes and diversity orders for any SR-CR pair on the Pareto boundary, as outlined in Table \ref{table1}. Consequently, the achievable SR-CR region of the CAPA-based ISAC system can be characterized as follows:
\begin{equation}\label{do_capa_region}
\mathcal{C}=\left\{\left({\mathcal{R}}_{\mathrm{s}},{\mathcal{R}}_{\mathrm{c}}\right)\left|{\mathcal{R}}_{\mathrm{s}}\in[0,\mathcal{R}_{\mathrm{s}}^{\tau}],
{\mathcal{R}}_{\mathrm{c}}\in[0,\overline{\mathcal{R}}_{\mathrm{c}}^{\tau}],\tau\in\left[0,1\right]\right.\right\}.
\end{equation}

\begin{table}[!t]
	\centering
	\caption{High-SNR Slopes ($\mathcal{S}$) and Diversity Orders ($\mathcal{D}$)}
	\begin{tabular}{|c|c|c|c|}\hline
		\multicolumn{1}{|c|}{\multirow{2}{*}{System design}} & SR & \multicolumn{2}{c|}{ECR}  \\ \cline{2-4}
		& $\mathcal{S}$  & $\mathcal{S}$ & $\mathcal{D}$ \\ \hline
		CAPA-ISAC S-C & $1/L$ & $1$  & $1$  \\ \hline
		CAPA-ISAC C-C & $1/L$ & $1$  & $\mathsf{DoF}$  \\ \hline
		CAPA-ISAC Pareto-optimal & $1/L$ & $1$  & $[1,\mathsf{DoF}]$  \\ \hline
		SPDA ($d_\mathrm{s}\leq\lambda/2$) S-C & $1/L$ & $1$  & $1$  \\ \hline
		SPDA ($d_\mathrm{s}\leq\lambda/2$) C-C & $1/L$ & $1$  & $\mathsf{DoF}$  \\ \hline
		SPDA ($d_\mathrm{s}>\lambda/2$) S-C & $1/L$ & $1$  & $1$  \\ \hline
		SPDA ($d_\mathrm{s}>\lambda/2$) C-C & $1/L$ & $1$  & $L_\mathrm{t}/d_\mathrm{s}$ \\ \hline
		FDSAC & $\kappa/L$ & $1-\kappa$  & $\mathsf{DoF}$  \\ \hline
	\end{tabular}	
	\vspace{2pt}
	\label{table1}
\end{table}

\begin{figure} [!t]
	\centering
	\includegraphics[height=2.1in]{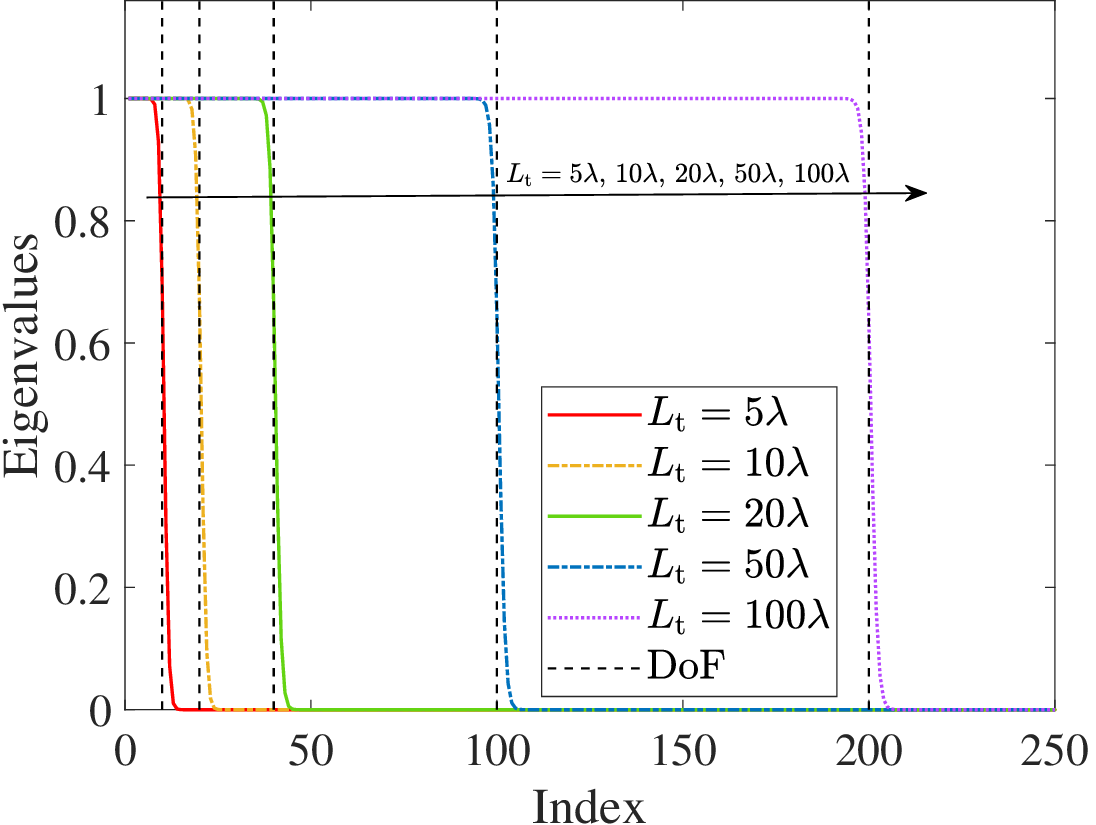}
	\caption{Illustration of the eigenvalues of $K(z,z')$.}
	\vspace{-5pt}
	\label{fig_eigen}
\end{figure}

\vspace{-3pt}
\section{Numerical Results}\label{section_numerical}
In this section, we present numerical results to evaluate the performance of CAPA-based ISAC systems and to validate the accuracy of the derived analytical expressions. Unless otherwise specified, the simulation parameters are set as follows: the wavelength is given by $\lambda=0.125$ m, $L_\mathrm{t}=L_\mathrm{r}=10 \lambda$, $d=2 \lambda$, $\overline{\gamma}_{\mathrm{s}}=\overline{\gamma}_{\mathrm{c}}=50$ dB, $L=4$, $\alpha_{\rm{s}}=1$, and $R_0=5$ bps/Hz. The target and CU are located at $\mathbf{p}_\mathrm{s}=[2 \, \text{m},1\,\text{m},1\,\text{m}]$ and $\mathbf{p}_\mathrm{c}=[4 \,\text{m},0\,\text{m},0\,\text{m}]$, respectively.

\begin{figure} [!t]
	\centering
	\includegraphics[height=2.1in]{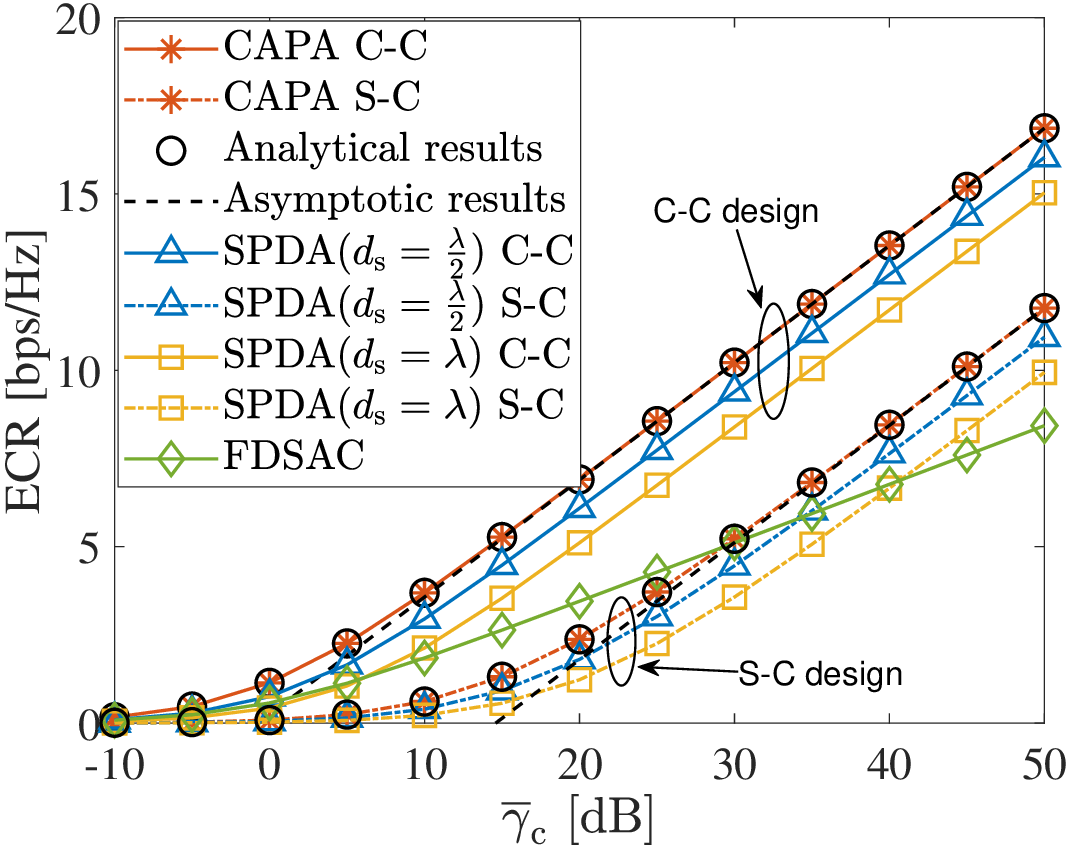}
	\caption{ECR versus transmit SNR.}
	\vspace{-5pt}
	\label{fig_ecr}
\end{figure}

\begin{figure} [!t]
	\centering
	\includegraphics[height=2in]{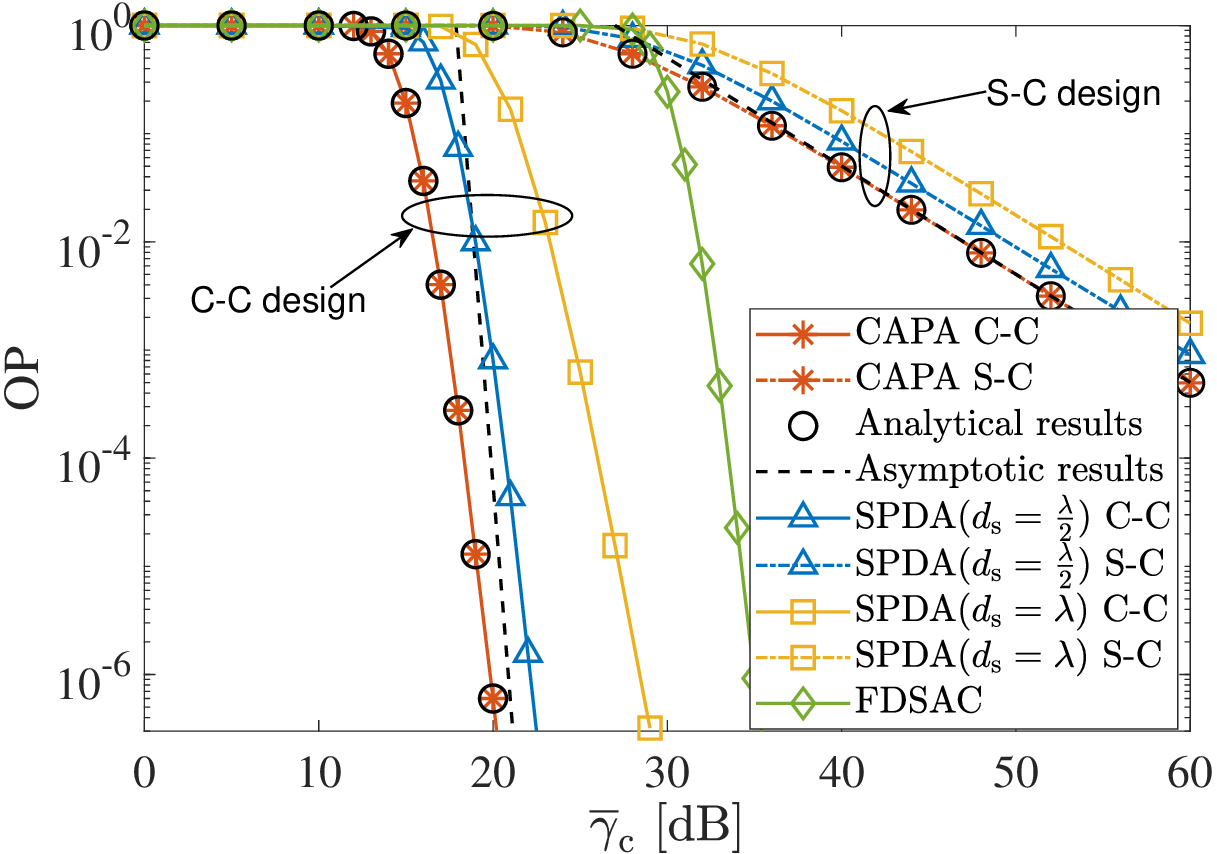}
	\caption{OP versus transmit SNR.}
	\vspace{-7pt}
	\label{fig_op}
\end{figure}

In Fig. \ref{fig_eigen}, we illustrate the ordered eigenvalues of $K(z,z')$ for different array lengths $L_\mathrm{t}$, computed using the method in Appendix \ref{eigenvalue}. The complexity-vs-accuracy trade-off parameter is set to $T=1000$. It can be observed from the figure that the eigenvalues of $K(z,z')$ exhibit a step-like behavior: for indexes below the effective DoF, i.e., $\frac{2L_\mathrm{t}}{\lambda}$, the eigenvalues decay slowly and are nearly one; in contrast, when the index is larger than the DoF,
the eigenvalues drop rapidly towards zero. This behavior is consistent with the statements in Remark~\ref{rem_landau}. Notably, this step-like behavior is already pronounced for moderate array sizes, e.g., $L_\mathrm{t}=10\lambda$. This observation underpins our earlier analysis of the communication channel gain, which models it as a finite weighted sum of exponentially distributed random variables. 

In the following results, we consider two baseline schemes for comparison:
\begin{itemize}
	\item \textbf{SPDA-based ISAC:} Both the transmit and receive arrays are SPDAs, where each element has length $\sqrt{\frac{\lambda^2}{4\pi}}$. The inter-element spacing is denoted by $d_{\mathrm{s}}$.
    \item \textbf{CAPA-based FDSAC:} For FDSAC, a fraction $\kappa=0.5$ of the total bandwidth and a fraction $\iota =0.5$ of the transmit power are allocated exclusively to sensing, while the remaining fractions are used for communications. In this case, the SR and ECR are given by $\mathcal{R} _{\mathrm{s}}^{\mathrm{f}}=\frac{\kappa}{L}\log _2\left( 1+\frac{\iota}{\kappa}\overline{\gamma }_{\mathrm{s}}L\alpha _{\mathrm{s}}G_{\mathrm{r}}G_{\mathrm{t}} \right) $ and $\overline{\mathcal{R} }_{\mathrm{c}}^{\mathrm{f}}=\mathbb{E} \left\{ \left( 1-\kappa \right) \log _2\left( 1+\frac{1-\iota}{1-\kappa}\overline{\gamma }_{\mathrm{c}}g_{\mathrm{c}} \right) \right\} $, respectively. Moreover, the achievable SR-CR region of FDSAC is characterized as 
    \begin{align}\label{do_fd_region}
    \mathcal{C} _{\mathrm{f}}=\left\{ ( \mathcal{R} _{\mathrm{s}},\mathcal{R} _{\mathrm{c}} ) \left| \begin{matrix}
	\mathcal{R} _{\mathrm{s}}\in [ 0,\mathcal{R} _{\mathrm{s}}^{\mathrm{f}} ] ,\mathcal{R} _{\mathrm{c}}\in [ 0,\overline{\mathcal{R}} _{\mathrm{c}}^{\mathrm{f}} ] ,\\
	\kappa \in [ 0,1 ] ,\iota \in [ 0,1 ]
\end{matrix} \right.\right\}.
    \end{align}
\end{itemize}

Fig. \ref{fig_ecr} and \ref{fig_op} plot the ECR and OP versus
the communication transmit SNR $\overline{\gamma}_\mathrm{c}$, respectively. It can be seen that the derived analytical results match well with the Monte Carlo simulation results, and the asymptotic results accurately track the simulation results in the high-SNR regime. As expected, the C-C design for CAPA-ISAC achieves the best communication performance. 
For the ECR, the C-C and S-C designs exhibit the
same high-SNR slope, which is higher than that of FDSAC. On the other hand, for outage performance, the diversity orders achieved by the C-C design is significantly higher than that of the S-C design. These observations corroborate the conclusion in Remark \ref{compare}.

Additionally, we can observe that CAPA consistently outperforms the conventional SPDA. In particular, CAPA attains a diversity order comparable to that of SPDA with $d_\mathrm{s}=\frac{\lambda}{2}$, while outperforming SPDA with $d_\mathrm{s}=\lambda$. This can be explained by the effective spatial DoFs. When the SPDA elements are spaced at sub-half-wavelength intervals, the number of the elements is no smaller than $\frac{2L_\mathrm{t}}{\lambda}=\mathsf{DoF}$, and the rank of the channel correlation matrix reaches $\mathsf{DoF}$. Consequently, SPDA is able to maintain the same DoFs and diversity gain as CAPA. However, when the spacing exceeds half a wavelength, the number of SPDA elements (and hence the rank of the correlation matrix) falls below $\mathsf{DoF}$, resulting in reduced DoF and diversity gain compared with CAPA. For ease of comparison, the high-SNR slopes and diversity orders achieved by each design and baseline scheme are summarised in Table~\ref{table1}.
 
\begin{figure} [!t]
	\centering
	\includegraphics[height=2.1in]{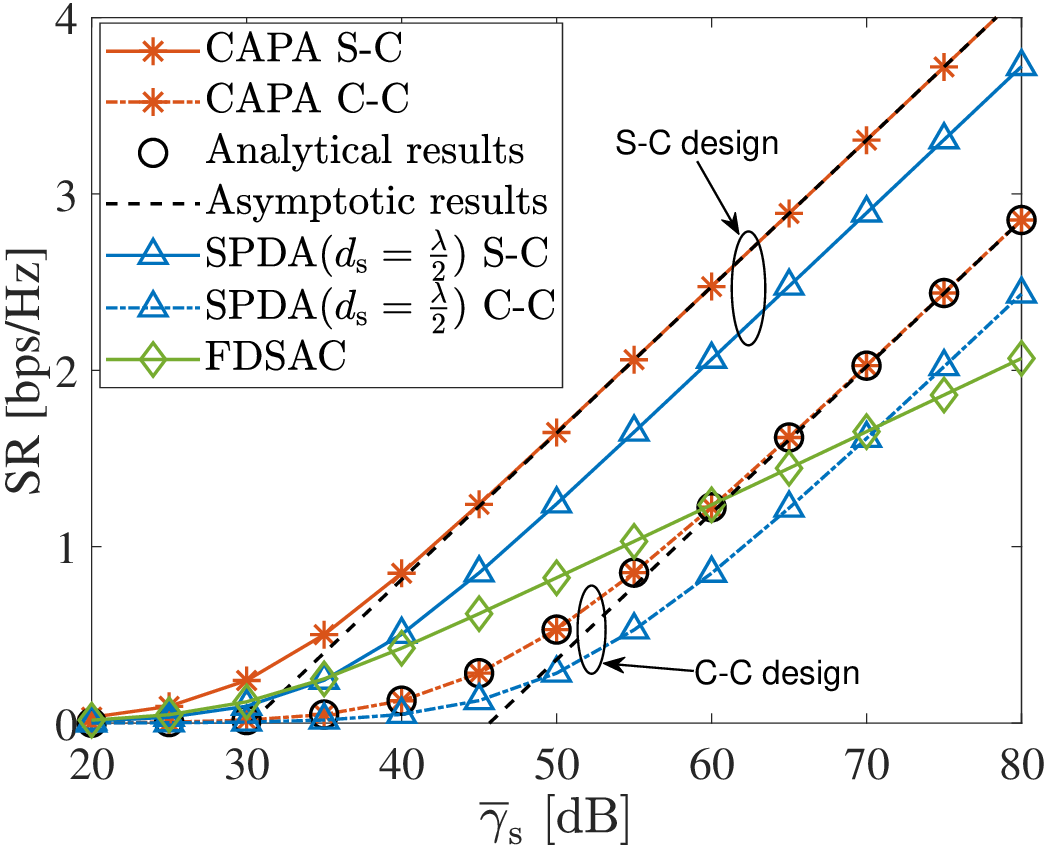}
	\caption{SR versus transmit SNR.}
	\vspace{-2pt}
	\label{fig_sr}
\end{figure}

\begin{figure} [!t]
	\centering
	\includegraphics[height=2.1in]{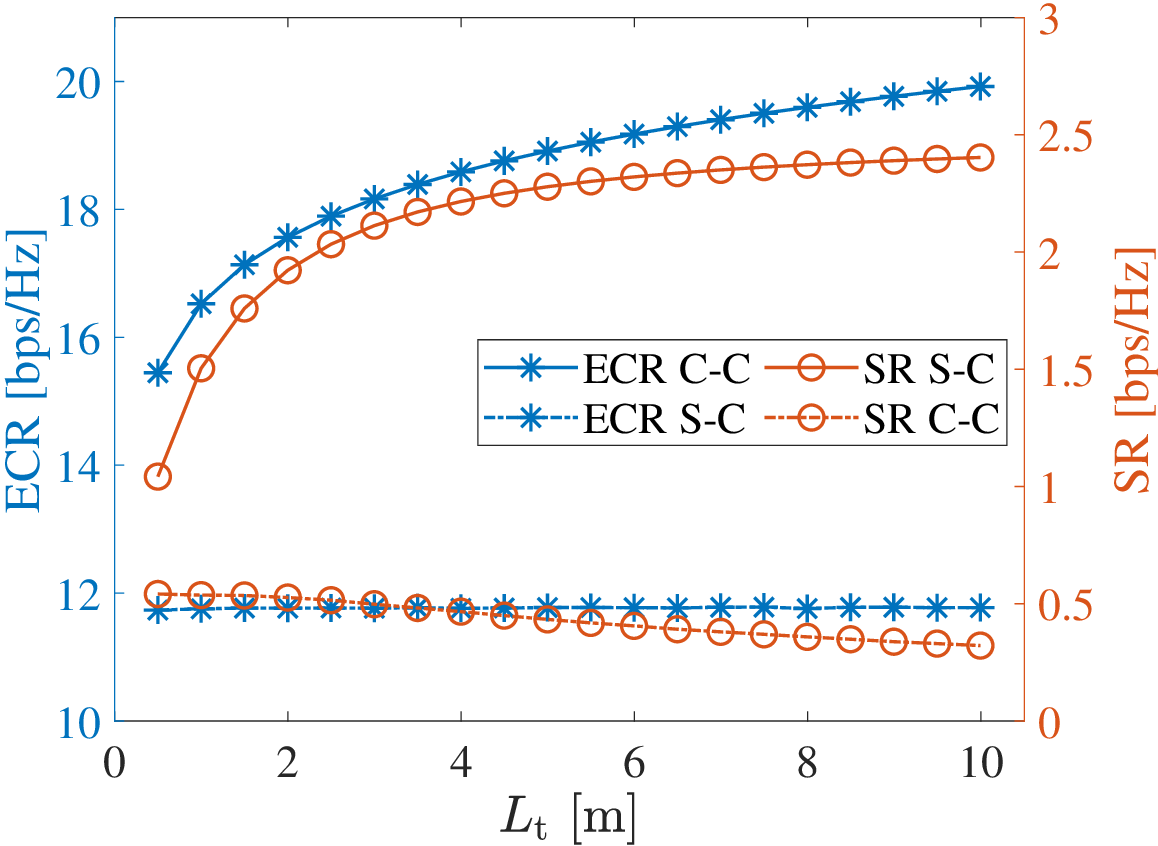}
	\caption{Rates versus transmit aperture length.}
	\vspace{-5pt}
	\label{fig_L}
\end{figure}
 
Fig. \ref{fig_sr} plots the SR versus the sensing SNR $\overline{\gamma}_\mathrm{s}$,
validating the accuracy of both the analytical and asymptotic results. As can be seen, CAPA-based S-C ISAC achieves the best sensing performance. Importantly, CAPA attains a higher SR than SPDA, while both schemes exhibit the same high-SNR slope. Additionally, in the high-SNR regime, the SR curves of ISAC increase more steeply than that of FDSAC, indicating that ISAC provides a larger sensing DoF than conventional FDSAC.
 
 Fig. \ref{fig_L} plots the ECR and SR achieved by CAPA as functions of $L_\mathrm{t}$. It is worth noting that the ECR under the S-C design exhibits weak sensitivity to the aperture length. This is because the closed-form results in \eqref{ins_sc_cr} shows that ECR depends on $L_\mathrm{t}$ primarily through the ratio $\frac{G_\mathrm{t}}{\Xi}$, which is the variance of $\rho\triangleq\int_{\mathcal{A} _{\mathrm{t}}}{h_{\mathrm{c}}}(\mathbf{t})h_{\mathrm{s}}^{*}(\mathbf{t})\mathrm{d}\mathbf{t}$. Enlarging $L_\mathrm{t}$ scales both $G_\mathrm{t}$ and $\Xi$ in a similar order, so $\frac{G_\mathrm{t}}{\Xi}$ remains approximately constant, leading to only marginal changes in ECR with $L_\mathrm{t}$. This is also reflected in the high-SNR approximation in \eqref{asy_sc_cr}, where the ECR is dominated by $-\log _2\frac{G_{\mathrm{t}}}{\Xi}$. Moreover, increasing the transmit aperture length can even decrease the SR under the C-C design. From \eqref{def_ins_cc_sr}, the SR depends on the ratio $\frac{\lvert \rho\rvert^2}{g_\mathrm{c}}$. When $L_\mathrm{t}$ grows, the C–C beamformer becomes increasingly concentrated on maximizing the communication link, which typically boosts $g_\mathrm{c}$ more rapidly than $\lvert \rho\rvert^2$, thereby decreasing the SR. Overall, these results indicate that aperture enlargement yields clear benefits only when the beamformer remains aligned with the target metric, which motivates the proposed Pareto trade-off design that balances the S-C and C-C beamformers.
 
 \begin{figure} [!t]
 	\centering
 	\includegraphics[height=2.1in]{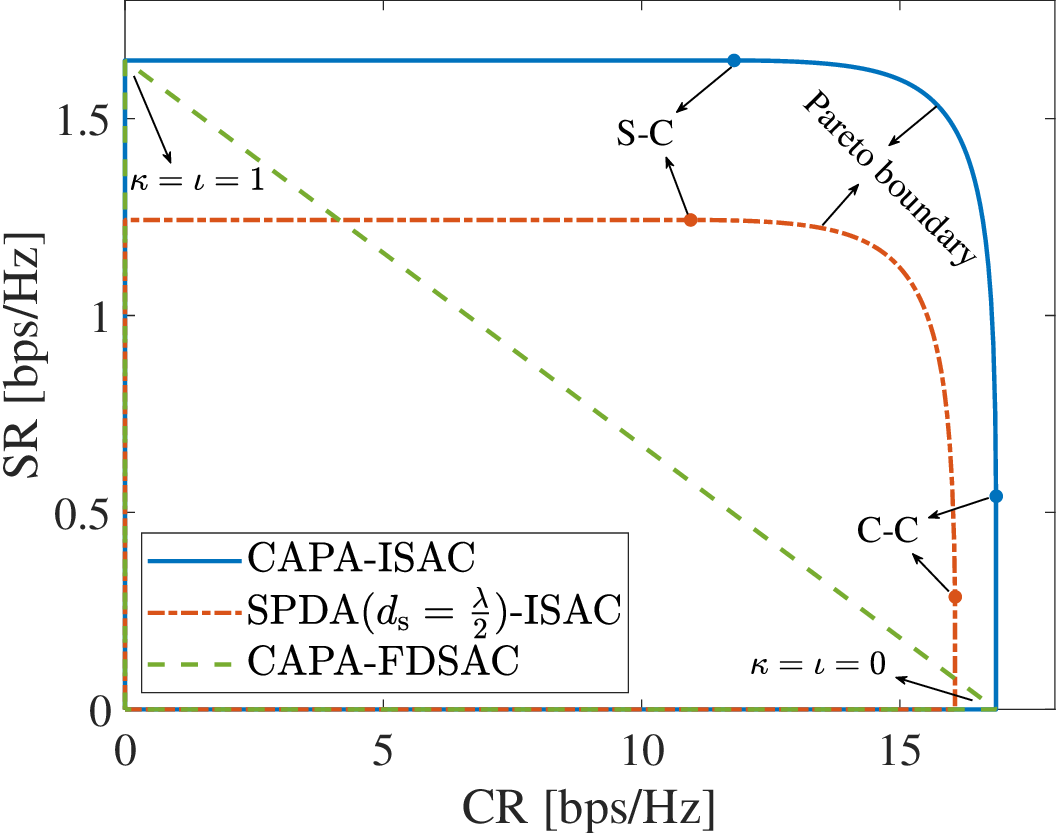}
 	\caption{Achievable SR-CR regions.}
 	\vspace{-5pt}
 	\label{fig_region}
 \end{figure}
 
{\figurename} {\ref{fig_region}} shows the achievable SR-CR regions achieved by three schemes: CAPA-based ISAC (defined in \eqref{do_capa_region}), SPDA-based ISAC, and CAPA-based FDSAC (defined in \eqref{do_fd_region}). For the ISAC scheme, two corner points are highlighted, corresponding to the S-C and C-C designs. By continuously adjusting the trade-off parameter $\tau$ from 1 to 0, we obtain a curve connecting these points, which forms the Pareto boundary of the ISAC rate region. A key observation is that the regions attainable by both SPDA-ISAC and CAPA-FDSAC lie strictly inside the CAPA-ISAC region, demonstrating the advantage of the proposed CAPA-based ISAC framework in jointly improving and balancing sensing and communication performance.



\section{Conclusion}\label{section_conclusion}
This article developed an operator-based framework for CAPA-enabled ISAC under a fading communication channel. By modeling the communication link as a continuous random field and leveraging the spectral structure of its correlation operator (via Landau’s eigenvalue theorem), we obtained a tractable characterization of the channel-gain statistics and, in turn, derived closed-form expressions for key metrics including SRs, ECRs, and OPs under the S–C and C–C designs, together with interpretable high-SNR asymptotics that reveal multiplexing and diversity behaviors. We further characterized the Pareto boundary of the achievable SR–CR region. Numerical results validate the analysis and demonstrate that the achievable SR–CR region of CAPA-based ISAC strictly contains those of both SPDA-based ISAC and CAPA-based FDSAC, confirming CAPA’s advantage in jointly improving and balancing sensing and communication performance. 

\begin{appendix}
	\setcounter{equation}{0}
	\renewcommand\theequation{A\arabic{equation}}

\subsection{Proof of Lemma \ref{lem_Rz}}\label{proof_lem_Rz}
We first calculate the autocorrelation of $\hat{H}\left( \boldsymbol{\kappa } \right) $. We define $\frac{\mathrm{e}^{-\mathrm{j}\mathbf{k}^{\mathsf{T}}\mathbf{p}_{\mathrm{c}}}}{\sqrt{\gamma \left( k_x,k_z \right)}}\triangleq \zeta\left( \mathbf{k} \right) $, which yields
\begin{align}
	&~\mathbb{E}\left\{ \hat{H}( \boldsymbol{\kappa } ) \hat{H}^*( \boldsymbol{\kappa }^{\prime} ) \right\} \notag\\
	&=\mathbb{E}\left\{ \iint_{\mathcal{D} ( \mathbf{k} )}\!{\zeta( \mathbf{k} ) W( \mathbf{k},\boldsymbol{\kappa } ) \mathrm{d}\mathbf{k}}\iint_{\mathcal{D} ( \mathbf{k}' )}\!{\zeta^*( \mathbf{k}^{\prime} ) W^*( \mathbf{k}^{\prime},\boldsymbol{\kappa }^{\prime} ) \mathrm{d}\mathbf{k}^{\prime}} \!\right\} \notag\\
	&=\iint_{\mathcal{D} ( \mathbf{k} )}\!{\iint_{\mathcal{D} ( \mathbf{k}' )}\!{\zeta( \mathbf{k} ) \zeta^*( \mathbf{k}^{\prime} )}}\mathbb{E}\left\{ W( \mathbf{k},\boldsymbol{\kappa } ) W^*( \mathbf{k}^{\prime},\boldsymbol{\kappa }^{\prime} ) \right\} \mathrm{d}\mathbf{k}^{\prime}\mathrm{d}\mathbf{k}\notag\\
	&=\iint_{\mathcal{D} ( \mathbf{k} )}{\iint_{\mathcal{D} ( \mathbf{k}' )}{\zeta( \mathbf{k} ) \zeta^*( \mathbf{k}^{\prime} )}}\delta ( \mathbf{k}-\mathbf{k}^{\prime} ) \delta ( \boldsymbol{\kappa }-\boldsymbol{\kappa }^{\prime} ) \mathrm{d}\mathbf{k}^{\prime}\mathrm{d}\mathbf{k}\notag\\
	&=\delta ( \boldsymbol{\kappa }-\boldsymbol{\kappa }^{\prime} ) \iint_{\mathcal{D} ( \mathbf{k} )}{\left| \zeta( \mathbf{k} ) \right|^2}\mathrm{d}\mathbf{k}=2\pi k_0\delta ( \boldsymbol{\kappa }-\boldsymbol{\kappa }^{\prime} ) .
\end{align}
On this basis, we can calculate the autocorrelation of $\hat{H}_z( \kappa _z )$ as follows:
\begin{align}
	&~\mathbb{E}\left\{ \hat{H}_z( \kappa _z ) \hat{H}_{z}^{*}( \kappa _{z}^{\prime} ) \right\}\notag\\ &=\int_{-\sqrt{k_{0}^{2}-\kappa _{z}^{2}}}^{\sqrt{k_{0}^{2}-\kappa _{z}^{2}}}{\int_{-\sqrt{k_{0}^{2}-{\kappa _{z}^{\prime}}^2}}^{\sqrt{k_{0}^{2}-{\kappa _{z}^{\prime}}^2}}{\frac{\mathbb{E}\left\{ \hat{H}\left( \boldsymbol{\kappa } \right) \hat{H}^*\left( \boldsymbol{\kappa }^{\prime} \right) \right\}}{\sqrt{\gamma \left( \kappa _x,\kappa _z \right) \gamma \left( \kappa _{x}^{\prime},\kappa _{z}^{\prime} \right)}}}}\mathrm{d}\kappa _{x}^{\prime}\mathrm{d}\kappa _x\notag\\
	&=2\pi k_0\!\int_{-\sqrt{k_{0}^{2}-\kappa _{z}^{2}}}^{\sqrt{k_{0}^{2}-\kappa _{z}^{2}}}{\int_{-\sqrt{k_{0}^{2}-{\kappa _{z}^{\prime}}^2}}^{\sqrt{k_{0}^{2}-{\kappa _{z}^{\prime}}^2}}{\frac{\delta ( \kappa _x\!-\!\kappa _{x}^{\prime} ) \delta ( \kappa _z\!-\!\kappa _{z}^{\prime} )}{\sqrt{\gamma( \kappa _x,\kappa _z ) \gamma ( \kappa _{x}^{\prime},\kappa _{z}^{\prime} )}}}}\mathrm{d}\kappa _{x}^{\prime}\mathrm{d}\kappa _x \notag\\
	&=2\pi k_0\delta \left( \kappa _z-\kappa _{z}^{\prime} \right) \int_{-\sqrt{k_{0}^{2}-\kappa _{z}^{2}}}^{\sqrt{k_{0}^{2}-\kappa _{z}^{2}}}{\frac{1}{\sqrt{\gamma \left( \kappa _x,\kappa _z \right) \gamma \left( \kappa _x,\kappa _{z}^{\prime} \right)}}\mathrm{d}\kappa _x}\notag\\
	&=2\pi ^2k_0\delta \left( \kappa _z-\kappa _{z}^{\prime} \right) .
\end{align}
Finally, the autocorrelation function of $\hat{h}_{\rm{c}}( z )$ can be calculated as follows:
\begin{align}
	&~R_{\rm{c}}\left( z,z^{\prime} \right) =\mathbb{E}\left\{ \hat{h}_{\rm{c}}\left( z \right) \hat{h}_{\rm{c}}^*\left( z^{\prime} \right) \right\}\notag\\ &=\frac{A_{\mathrm{s}}^{2}( k_0 )}{16\pi ^4}\!\int_{-k_0}^{k_0}\!{\int_{-k_0}^{k_0}\!{\mathbb{E}\left\{ \hat{H}_z( \kappa _z ) \hat{H}_{z}^{*}( \kappa _{z}^{\prime} ) \right\}}\mathrm{e}^{\mathrm{j}\kappa _zz}}\mathrm{e}^{-\mathrm{j}\kappa _{z}^{\prime}z^{\prime}}\mathrm{d}\kappa _z\mathrm{d}\kappa _{z}^{\prime}\notag\\
	&=\frac{1}{2k_0}\int_{-k_0}^{k_0}{\int_{-k_0}^{k_0}{\delta \left( \kappa _z-\kappa _{z}^{\prime} \right)}\mathrm{e}^{\mathrm{j}\kappa _zz}}\mathrm{e}^{-\mathrm{j}\kappa _{z}^{\prime}z^{\prime}}\mathrm{d}\kappa _z\mathrm{d}\kappa _{z}^{\prime}\notag\\
	&=\frac{1}{2k_0}\int_{-k_0}^{k_0}{\mathrm{e}^{\mathrm{j}\kappa _z\left( z-z^{\prime} \right)}\mathrm{d}\kappa _z}=\frac{\sin \left( k_0(z-z^\prime) \right)}{k_0(z-z^\prime)}.
\end{align}

\subsection{Numerical Evaluation for Eigenvalues of $R_{\rm{c}}( z,z^{\prime} )$}\label{eigenvalue}
The characteristic equation for  $R_{\rm{c}}( z,z^{\prime} )$ is given by
\begin{equation}
	\int_{\mathcal{Z}}{R_{\mathrm{c}}(z,z^{\prime})\phi \left( z^{\prime} \right)}\mathrm{d}z^{\prime}=\lambda \phi \left( z \right) .
\end{equation}
In order to discretize the integral, we use the Gauss-Legendre quadrature given as follows \cite{math}:
\begin{equation}
	\int_a^b{f\left( x \right)}\mathrm{d}x\approx \frac{b-a}{2}\sum_{t=1}^T{\omega _tf\left( \frac{b-a}{2}\eta _t+\frac{a+b}{2} \right)},
\end{equation}
where $T$ is a complexity-vs-accuracy trade-off parameter, $\omega _t$ are the quadrature weights, and $\eta _m$ are the roots of the $T$-th Legendre polynomial \cite{math}. Accordingly, we have 
\begin{equation}
	\int_{\mathcal{Z}}{R_{\mathrm{c}}(z,z^{\prime})\phi \left( z^{\prime} \right)}\mathrm{d}z^{\prime}\approx \sum_{t=1}^T{\frac{L_{\mathrm{t}}}{2}\omega _tR_{\mathrm{c}}\left( z,z_t \right) \phi \left( z_t \right)}=\lambda \phi \left( z \right), 
\end{equation}
where $z_t=\frac{L_{\mathrm{t}}\eta _t+d+L_{\mathrm{t}}}{2}$ for $t\in \left\{ 1,....,T \right\} $. We also take $z$ as the same set of nodes $\left\{ z_{t^{\prime}} \right\} _{t^{\prime}=1}^{T}$, which yields
\begin{equation}
\sum\nolimits_{t=1}^T{\frac{L_{\mathrm{t}}}{2}\omega _tR_{\mathrm{c}}\left( z_{t^{\prime}},z_t \right) \phi \left( z_t \right)}=\lambda \phi \left( z_{t^{\prime}} \right) .
\end{equation} 
This can be equivalently written as 
\begin{equation}
	\mathbf{Z}\boldsymbol{\phi }=\lambda \boldsymbol{\phi },
\end{equation} 
where $\left[ \mathbf{Z} \right] _{t^{\prime},t}=\frac{L_{\mathrm{t}}}{2}\omega _tR_{\mathrm{c}}\left( z_{t^{\prime}},z_t \right) $, and $\boldsymbol{\phi }=\left[ \phi \left( z_1 \right) ,...,\phi \left( z_T \right) \right] ^{\mathsf{T}}$. Consequently, the eigenvalues of the matrix $\mathbf{Z}$ are approximations of the eigenvalues of $R_{\mathrm{c}}\left( z_{t^{\prime}},z_t \right) $, and the accuracy improves as $T$ increases.
\subsection{Proof of Lemma \ref{lem_h_sta}}\label{proof_lem_h_sta}
Given $\Psi  _n\sim \mathcal{C} \mathcal{N} \left( 0,1 \right) $ for $\forall n$, it is clear that $\sum\nolimits_{n=1}^{\infty}{\sqrt{\lambda _n}}\phi _n\left( z \right) \Psi  _n \triangleq\bar{h}\left( z \right) $ is a zero-mean Gaussian random field. The autocorrelation function of $\bar{h}\left( z \right) $ is given by
\begin{equation}
\begin{split}
&\mathbb{E}\left\{ \bar{h}\left( z \right) \bar{h}^*( z^{\prime} ) \right\}\\ &=\sum\nolimits_{n=1}^{\infty}{\sum\nolimits_{n^{\prime}=1}^{\infty}{\sqrt{\lambda _n}\lambda _{n^{\prime}}}}\phi _n( z ) \phi _{n^{\prime}}^{*}( z^{\prime} ) \mathbb{E}\left\{ \Psi  _n\Psi  _{n^{\prime}}^{*} \right\}.
\end{split}
\end{equation}
Since $\left\{ \Psi  _n \right\} _{n=1}^{\mathrm{inf}}$ are i.i.d. ZUCG, we have $E\left\{ \Psi  _n\Psi  _{n^{\prime}}^{*} \right\} =\delta _{n,n^{\prime}}$, which yields
\begin{equation}
	E\left\{ \bar{h}( z ) \bar{h}^*( z^{\prime} ) \right\} =\sum\nolimits_{n=1}^{\infty}{\lambda _n\phi _n( z ) \phi _{n}^{*}( z^{\prime} )}=R_{\rm{c}}( z,z^{\prime} ). 
\end{equation} 
This means that $\bar{h}\left( z \right) $ also has the same autocorrelation with $\hat{h}_{\rm{c}}\left( z \right) $, which completes the proof.
	 
\subsection{Proof of Theorem \ref{the_ins_sc_sr}}\label{proof_the_ins_sc_sr}
The distance between the point $\mathbf{t}=\left[ 0,0,z \right] \in \mathcal{A} _{\mathrm{t}}$ and the target position ${\mathbf{p}}_{\rm{s}}=[p_{\rm{s},x},p_{\rm{s},y},p_{\rm{s},z}]$ is given by
\begin{align}\label{distance}
\left\| \mathbf{t}-\mathbf{p}_{\mathrm{s}} \right\| =\sqrt{z^2-2p_{\mathrm{s},z}z+p_{\mathrm{s},x}^{2}+p_{\mathrm{s},y}^{2}+p_{\mathrm{s},z}^{2}}.	
\end{align}
Therefore, the sensing link $h_{\mathrm{s}}(\mathbf{t})$ can be expressed as follows:
\begin{equation}\label{hshat}
h_{\mathrm{s}}(\mathbf{t})=\frac{\mathrm{e}^{-\mathrm{j}k_0\sqrt{z^2-2p_{\mathrm{s},z}z+p_{\mathrm{s},x}^{2}+p_{\mathrm{s},y}^{2}+p_{\mathrm{s},z}^{2}}}}{\sqrt{4\pi \left( z^2-2p_{\mathrm{s},z}z+p_{\mathrm{s},x}^{2}+p_{\mathrm{s},y}^{2}+p_{\mathrm{s},z}^{2} \right)}}\triangleq\hat{h}_{\mathrm{s}}\left( z \right) .
\end{equation}
As a result, we have
\begin{equation}
\int_{\mathcal{A} _{\mathrm{t}}}\!\!{\left| h_{\mathrm{s}}(\mathbf{t}) \right|^2}\mathrm{d}\mathbf{t}\!=\!\frac{1}{4\pi}\!\int_{d/2}^{d/2+L_{\mathrm{t}}}\!\!{\frac{1}{z^2\!-\!2p_{\mathrm{s},z}z\!+\!p_{\mathrm{s},x}^{2}\!+p_{\mathrm{s},y}^{2}\!+p_{\mathrm{s},z}^{2}}\mathrm{d}z}.
\end{equation}
The above integral can be calculated with the aid of \cite[Eq. (2.103.4)]{integral}, which yields $\int_{\mathcal{A} _{\mathrm{t}}}{\left| h_{\mathrm{s}}(\mathbf{t}) \right|^2}\mathrm{d}\mathbf{t}=G_{\mathrm{t}}$. Similarly, we can obtain $\int_{\mathcal{A} _{\mathrm{r}}}{\left| h_{\mathrm{s}}(\mathbf{r}) \right|^2}\mathrm{d}\mathbf{r}=G_{\mathrm{r}}$.
Inserting them into \eqref{ins_sc_sr_def} gives \eqref{ins_sc_sr}.

\subsection{Proof of Theorem \ref{the_ins_sc_cr} and \ref{the_inc_sc_op}}\label{proof_the_ins_sc_cr}
It is straightforward to show that $\rho\sim \mathcal{C} \mathcal{N} \left( 0,\Xi \right) $. Hence, $\left| \rho \right|^2$ follows an exponential distribution, i.e., $\left| \rho \right|^2\sim \mathrm{Exp}\left( \frac{1}{\Xi} \right) $. Consequently, the ECR is calculated as
\begin{equation}\label{a16}
		\overline{\mathcal{R} }_{\mathrm{c}}^{\mathrm{s}}=\int_0^{\infty}{\log _2\left( 1+\overline{\gamma }_{\mathrm{c}}\frac{x}{\int_{\mathcal{A} _{\mathrm{t}}}{\left| h_{\mathrm{s}}(\mathbf{t}) \right|^2}\mathrm{d}\mathbf{t}} \right)}\frac{1}{\Xi}\mathrm{e}^{-\frac{x}{\Xi}}\mathrm{d}x.	
\end{equation}
Inserting $\int_{\mathcal{A} _{\mathrm{t}}}{\left| h_{\mathrm{s}}(\mathbf{t}) \right|^2}\mathrm{d}\mathbf{t}=G_{\mathrm{t}}$ into \eqref{a16} and using \cite[(4.337.2)]{integral}, the results of \eqref{ins_sc_cr} can be derived.

The OP for the S-C design is calculated as 
\begin{align}
	\mathcal{P} _{\mathrm{s}}&=\mathrm{Pr}\left( \mathcal{R} _{\mathrm{c}}^{\mathrm{s}}<\mathcal{R} _0 \right)\notag\\
	& =\mathrm{Pr}\left( \left| \rho \right|^2<\frac{2^{\mathcal{R} _0}-1}{\overline{\gamma }_{\mathrm{c}}}\int_{\mathcal{A} _{\mathrm{t}}}{\left| h_{\mathrm{s}}(\mathbf{t}) \right|^2}\mathrm{d}\mathbf{t} \right)\notag\\ &=F_{\rho}\left( \frac{2^{\mathcal{R} _0}-1}{\overline{\gamma }_{\mathrm{c}}}\int_{\mathcal{A} _{\mathrm{t}}}{\left| h_{\mathrm{s}}(\mathbf{t}) \right|^2}\mathrm{d}\mathbf{t} \right).  
\end{align}
Here, $F_{\rho}\left( x \right) =1-\mathrm{e}^{-\frac{x}{\Xi}}$ is the cumulative distribution function (CDF) of $\left| \rho \right|^2$. The results of \eqref{ins_sc_op} then follows immediately.

\subsection{Proof of Lemma \ref{lem_ins_cc_sr}}\label{proof_lem_ins_cc_sr}
We can rewrite \eqref{def_ins_cc_sr} as follows:
\begin{equation}\label{a18}
\mathcal{R} _{\mathrm{s}}^{\mathrm{c}}=\frac{1}{L}\log _2\Delta -\frac{1}{L}\log _2g_{\mathrm{c}} ,
\end{equation}
where 
\begin{align}
	\Delta &=\!\int_{\mathcal{Z}}{\lvert \hat{h}_{\mathrm{c}}(z) \rvert^2}\mathrm{d}z+L\alpha _{\mathrm{s}}\overline{\gamma }_{\mathrm{s}}\int_{\mathcal{A} _{\mathrm{r}}}\!\!{\lvert h_{\mathrm{s}}(\mathbf{r}) \rvert^2}\mathrm{d}\mathbf{r}\left| \int_{\mathcal{Z}}{\hat{h}_{\mathrm{s}}(z)}\hat{h}_{\mathrm{c}}^{*}(z)\mathrm{d}z \right|^2\notag\\
	&=\int_{\mathcal{Z}}\int_{\mathcal{Z}}\hat{h}_{\mathrm{c}}^{*}(z)\big( \delta \left( z-z^{\prime} \right) +\notag\\
	&~~~L\alpha _{\mathrm{s}}\overline{\gamma }_{\mathrm{s}}\int_{\mathcal{A} _{\mathrm{r}}}\!{\left| h_{\mathrm{s}}(\mathbf{r}) \right|^2}\mathrm{d}\mathbf{r}\hat{h}_{\mathrm{s}}\left( z \right) \hat{h}_{\mathrm{s}}^{*}\left( z^{\prime} \right) \Big) \hat{h}_{\mathrm{c}}(z^{\prime})\mathrm{d}z\mathrm{d}z^{\prime}.
\end{align}
Furthermore, it is easily shown that $\hat{h}_{\mathrm{c}}(z)$ can be written as $\hat{h}_{\mathrm{c}}(z)=\int_{\mathcal{Z}}{\tilde{R}_{\mathrm{c}}\left( z,z_1 \right) \overline{h}(z_1)\mathrm{d}z_1}$, where $\overline{h}\left( z \right) \sim \mathcal{C} \mathcal{N} \left( 0,1 \right) $ satisfies $\mathbb{E} \{\overline{h}\left( z \right) \overline{h}^*(z^{\prime})\}=\delta (z-z^\prime)$. On this basis, $\Delta$ can be expressed as 
\begin{equation}\label{a20}
	\Delta=\int_{\mathcal{Z}}{\int_{\mathcal{Z}}{\overline{h}^*(z_1)Q\left( z_1,z_2 \right) \overline{h}(z_2)\mathrm{d}z_1}\mathrm{d}z_2}.
\end{equation}
Substituting the eigendecomposition $Q\left( z_1,z_2 \right) =\sum\nolimits_{n=1}^{\infty}{\nu _n}\psi _n\left( z_1 \right) \psi _{n}^{*}\left( z_2 \right) $ into \eqref{a20} gives
\begin{align}
\Delta&=\sum\nolimits_{n=1}^{\infty}{\nu _n}\int_{\mathcal{Z}}{\int_{\mathcal{Z}}{\overline{h}^*(z_1)\psi _n( z_1 ) \psi _{n}^{*}( z_2 ) \overline{h}(z_2)\mathrm{d}z_1}\mathrm{d}z_2}\notag\\
&=\sum\nolimits_{n=1}^{\infty}{\nu _n}\left| \int_{\mathcal{Z}}{\overline{h}^*(z)\psi _n\left( z \right) \mathrm{d}z} \right|^2.
\end{align}
Since it is clear that $\int_{\mathcal{Z}}{\overline{h}^*(z)\psi _n\left( z \right) \mathrm{d}z}$ is a ZUCG random variable, based on the Landau’s eigenvalue theorem, we can write $\Delta=\sum_{n=1}^{\mathsf{DoF}_Q}{\nu _n\left| \Phi _n \right|^2}$, which, together with \eqref{a18}, leads to the results in Lemma~\ref{lem_ins_cc_sr}.

\subsection{Proof of Lemma \ref{Lemma_Subspace}}\label{proof_Lemma_Subspace}
We prove Lemma \ref{Lemma_Subspace} by contradiction. Assume that $w({\mathbf{t}})$ does not lie in the subspace spanned by $\{h_{\rm{s}}^{*}({\mathbf{t}}),h_{\rm{c}}^{*}(\mathbf{t})\}$. Under this assumption, $w({\mathbf{t}})$ must contain a component that is orthogonal to both $h_{\rm{s}}^{*}({\mathbf{t}})$ and $h_{\rm{c}}^{*}(\mathbf{t})$, which can be represented as follows:
\begin{align}
w({\mathbf{t}})=\hat{w}(\mathbf{t})+\hat{w}_{\bot}(\mathbf{t}),
\end{align}
where  $\hat{w}(\mathbf{t})=ah_{\rm{s}}^{*}({\mathbf{t}})+bh_{\rm{c}}^{*}(\mathbf{t})$ with $a,b\in{\mathbb{C}}$ is an arbitrary function that lies in the subspace spanned by $\{h_{\rm{s}}^{*}({\mathbf{t}}),h_{\rm{c}}^{*}(\mathbf{t})\}$, while the term $\hat{w}_{\bot}(\mathbf{t})$ represents the orthogonal component that lies outside this subspace. Consequently, we can obtain
\begin{subequations}\label{Subspace_Approach_Step2}
\begin{align}
\Upsilon_{\rm{s}}({{w}}({\mathbf{t}}))&=\left\lvert\int_{{\mathcal{A}}_{\rm{t}}}h_{\rm{s}}({\mathbf{t}})\hat{w}(\mathbf{t}){\rm{d}}{\mathbf{t}}\right\rvert^2
=\Upsilon_{\rm{s}}(\hat{w}(\mathbf{t})),\\
\Upsilon_{\rm{c}}({{w}}({\mathbf{t}}))&=\left\lvert\int_{{\mathcal{A}}_{\rm{t}}}h_{\rm{c}}({\mathbf{t}})\hat{w}(\mathbf{t}){\rm{d}}{\mathbf{t}}\right\rvert^2
=\Upsilon_{\rm{c}}(\hat{w}(\mathbf{t})).
\end{align}
\end{subequations} 
Additionally, the power constraint yields
\begin{align}\label{Subspace_Approach_Step3}
\int_{{\mathcal{A}}_{\rm{t}}}\lvert w({\mathbf{t}})\rvert^2{\rm{d}}{\mathbf{t}}=
\int_{{\mathcal{A}}_{\rm{t}}}\lvert \hat{w}(\mathbf{t})\rvert^2{\rm{d}}{\mathbf{t}}+
\int_{{\mathcal{A}}_{\rm{t}}}\lvert \hat{w}_{\bot}({\mathbf{t}})\rvert^2{\rm{d}}{\mathbf{t}}.
\end{align}
From \eqref{Subspace_Approach_Step2} and \eqref{Subspace_Approach_Step3}, it can be observed that the orthogonal component $\check{w}_{\bot}(\mathbf{t})$ neither enhances SR nor CR, while it still consumes extra power. As $\check{w}_{\bot}(\mathbf{t})$ provides no benefit to system performance but increases power usage, it should be eliminated. Consequently, the optimal beamformer can be expressed as $w({\mathbf{t}})=\hat{w}(\mathbf{t})=ah_{\rm{s}}^{*}({\mathbf{t}})+bh_{\rm{c}}^{*}(\mathbf{t})$, which lies within the subspace spanned by $\{h_{\rm{s}}^{*}({\mathbf{t}}),h_{\rm{c}}^{*}(\mathbf{t})\}$.
\subsection{Proof of Lemma \ref{lem_w_star}}\label{proof_lem_w_star}
The KKT conditions for problem \eqref{pareto_problme_2} is given as follows:
\begin{align}
	&\nabla(-\Upsilon)+\nu \nabla(\lVert{\mathbf{w}}\rVert^2-1)+\mu_1\nabla (\tau\Upsilon-\lvert{{\mathbf{h}}}_{\mathrm{s}}^{\mathsf{T}}\mathbf{w}\rvert^2)\notag\\
	&+\mu_2\nabla ((1-\tau)\Upsilon-\lvert{\mathbf{h}}_{\mathrm{c}}^{\mathsf{T}}\mathbf{w}\rvert^2)={\mathbf{0}}, \label{KKT_1} \\ 
	&\mu_1(\tau\Upsilon-\lvert{{\mathbf{h}}}_{\mathrm{s}}^{\mathsf{T}}\mathbf{w}\rvert^2)=0,\label{KKT_2} \\ &\mu_2((1-\tau)\Upsilon-\lvert{\mathbf{h}}_{\mathrm{c}}^{\mathsf{T}}\mathbf{w}\rvert^2)=0,\label{KKT_3}\\
	&\mu_1\geq0,\ \mu_2\geq0,\ \nu\in{\mathbb{R}},
\end{align}
where $\nu$, $\mu_1$, and $\mu_2$ are real-valued Lagrangian multipliers. We can obtain from \eqref{KKT_1} that
\begin{align}
		&\tau\mu_1+(1-\tau)\mu_2=1,\label{KKT_1_Dev2}\\		&(\mu_1{{\mathbf{h}}}_{\mathrm{s}}^*{{\mathbf{h}}}_{\mathrm{s}}^{\mathsf{T}}+\mu_2{\mathbf{h}}_{\mathrm{c}}^*{\mathbf{h}}_{\mathrm{c}}^{\mathsf{T}}){\mathbf{w}}=\nu{\mathbf{w}}.\label{KKT_1_Dev1}   
\end{align}
From \eqref{KKT_1_Dev2}, it can be observed that $\mu_1$ and $\mu_2$ cannot both be zero at the same time. Hence, we analyze the problem under the following three possible cases.
\subsubsection{$\mu_1>0$ and $\mu_2=0$} In this case, we can obtain from \eqref{KKT_2} that $\Upsilon=\frac{\lvert{\mathbf{h}}_{\mathrm{s}}^{\mathsf{T}}\mathbf{w}\rvert^2}{\tau}$. To maximize $\Upsilon$, the optimal solution is ${\mathbf{w}_\tau}=\frac{{\mathbf{h}}_{\mathrm{s}}^*}{\lVert{\mathbf{h}}_{\mathrm{s}}\rVert}=\frac{\mathbf{h}_{\mathrm{s}}^*}{\sqrt{G_{\mathrm{t}}}}$, leading to $\Upsilon_\star=\frac{G_{\mathrm{t}}}{\tau}$. By substituting this into the constraint $\lvert{\mathbf{h}}_{\mathrm{c}}^{\mathsf{T}}{\mathbf{w}}\rvert^2\geq(1-\tau)\Upsilon$ gives the range $\tau\in \left[ \frac{G_{\mathrm{t}}^{2}}{G_{\mathrm{t}}^{2}+\left| \rho  \right|^2},1 \right] $.
\subsubsection{$\mu_1=0$ and $\mu_2>0$} By following the same procedure as in above case, we derive ${\mathbf{w}_\tau}=\frac{{\mathbf{h}}_{\mathrm{c}}^*}{\sqrt{g_{\rm{c}}}}$ with $\tau\in \left[ 0,\frac{\left| \rho  \right|^2}{\left| \rho  \right|^2+g_{\mathrm{c}}^{2}} \right] $.
\subsubsection{$\mu_1>0$ and $\mu_2>0$}
From \eqref{KKT_2} and \eqref{KKT_3}, we have
\begin{equation}\label{KKT_1_Dev3}	\Upsilon=\frac{1}{\tau}\lvert{\mathbf{h}}_{\mathrm{s}}^{\mathsf{T}}\mathbf{w}\rvert^2=\frac{1}{1-\tau}\lvert{\mathbf{h}}_{\mathrm{c}}^{\mathsf{T}}{\mathbf{w}}\rvert^2.
\end{equation}
Based on \eqref{KKT_1_Dev1}, $\mathbf{w}$ can be expressed as a linear combination of $\mathbf{h}_{\mathrm{s}}$ and $\mathbf{h}_{\mathrm{c}}$:
\begin{align}\label{wab}
	\mathbf{w}=\frac{\mu _1\mathbf{h}_{\mathrm{s}}^{\mathsf{T}}\mathbf{w}}{\nu}\mathbf{h}_{\mathrm{s}}^*+\frac{\mu _2\mathbf{h}_{\mathrm{c}}^{\mathsf{T}}\mathbf{w}}{\nu}\mathbf{h}_{\mathrm{c}}^*\triangleq \alpha\mathbf{h}_{\mathrm{s}}^*+\beta\mathbf{h}_{\mathrm{c}}^*.
\end{align}
Subsequently, from \eqref{KKT_1_Dev3}, we can obtain
\begin{equation}\label{alphabeta}
	\frac{\alpha}{\beta}=\frac{\mu _1\sqrt{\tau}}{\mu _2\sqrt{1-\tau}\mathrm{e}^{-\mathrm{j}\angle \rho }}.
\end{equation}
Substituting \eqref{wab} into \eqref{KKT_1_Dev1} gives
\begin{align}\label{mu_eta}
	{\mu_1}(G_{\mathrm{t}}+\rho{\beta}/{\alpha})={\mu_2}(g_{\mathrm{c}}+\rho ^*{\alpha}/{\beta})=\nu,   
\end{align}
which, together with \eqref{KKT_1_Dev2} and \eqref{alphabeta}, yields 
\begin{align}
&\mu _1=\frac{g_{\mathrm{c}}-\sqrt{( 1-\tau ) /\tau}\left| \rho  \right|}{( 1-\tau ) G_{\mathrm{t}}+\tau g_{\mathrm{c}}-2\sqrt{\tau ( 1-\tau )}\left| \rho  \right|},\\
&\mu _2=\frac{G_{\mathrm{t}}-\sqrt{\tau /( 1-\tau )}\left| \rho  \right|}{( 1-\tau ) G_{\mathrm{t}}+\tau g_{\mathrm{c}}-2\sqrt{\tau ( 1-\tau )}\left| \rho  \right|}.
\end{align}
Under the conditions $\mu_1>0$ and $\mu_2>0$, $\tau$ should satisfies $\tau \in \left( \frac{\left| \rho  \right|^2}{\left| \rho  \right|^2+G_{\mathrm{t}}^{2}},\frac{g_{\mathrm{d}}^{2}}{g_{\mathrm{d}}^{2}+\left| \rho  \right|^2} \right) $. According to \eqref{alphabeta}, the optimal $\mathbf{w}$ in this scenario can be formulated as follows:
\begin{equation}
	\mathbf{w}_\tau=\frac{\mu _1\sqrt{\tau}}{\varsigma}\mathbf{h}_{\mathrm{s}}^*+\frac{\mu _2\sqrt{1-\tau}\mathrm{e}^{-\mathrm{j}\angle \rho}}{\varsigma}\mathbf{h}_{\mathrm{c}}^*,
\end{equation}
where $\varsigma$ is for normalization. Taken together, the proof of Lemma \ref{lem_w_star} is completed.
\end{appendix}

\bibliographystyle{IEEEtran}
\bibliography{IEEEabrv}

@STRING{IEEE_WM_COM        = "{IEEE} Wireless Commun."}

@STRING{IEEE_ANTEN        ="{IEEE} Trans. Antennas Propag."}

@book{integral,
	author = {I. S. Gradshteyn and I. M. Ryzhik}, 
	title = {Table of Integrals, Series and Products},
	publisher = {Academic Press},
	year = 2007,
	edition = 7,
	address = {New York, NY, USA}
}

@book{math,
	title={NIST Handbook of Mathematical Functions},
	author={Olver, Frank W and Lozier, Daniel W and Boisvert, Ronald F and Clark, Charles W},
	publisher={Cambridge univ. press},
	year=2010,
	address = {Cambridge, U.K.}
}

@article{poon2005degrees,
  title={Degrees of freedom in multiple-antenna channels: A signal space approach},
  author={Poon, Ada SY and Brodersen, Robert W and Tse, David NC},
  journal={IEEE Trans. Inf. Theory},
  volume={51},
  number={2},
  pages={523--536},
  year={2005},
  month=feb,
  publisher={IEEE}
}

@article{zwick2002stochastic,
  title={A stochastic multipath channel model including path directions for indoor environments},
  author={Zwick, Thomas and Fischer, Christian and Wiesbeck, Werner},
  journal={IEEE J. Sel. Areas Commun.},
  volume={20},
  number={6},
  pages={1178--1192},
  year={2002},
  month=aug,
  publisher={IEEE}
}

@book{richards2005fundamentals,
  title={Fundamentals of Radar Signal Processing},
  author={Richards, Mark A},
  year={2014},
  publisher={New York, NY, USA: McGraw-Hill Educ.}
}

@article{ouyang2024impact,
  title={On the impact of reactive region on the near-field channel gain},
  author={Ouyang, Chongjun and Wang, Zhaolin and Zhao, Boqun and Zhang, Xingqi and Liu, Yuanwei},
  journal={IEEE Commun. Lett.},
  volume={28},
  number={10},
  pages={2417--2421},
  year={2024},
  month=oct,
  publisher={IEEE}
}

@article{zhao2024continuous,
  title={Continuous aperture array ({CAPA})-based wireless communications: Capacity characterization},
  author={Zhao, Boqun and Ouyang, Chongjun and Zhang, Xingqi and Liu, Yuanwei},
  journal={IEEE Trans. Wireless Commun.},
   year={2025},
  volume={24},
  number={12},
  pages={10456-10473}
}

@article{ouyang2023integrated,
  title={Integrated sensing and communications: A mutual information-based framework},
  author={Ouyang, Chongjun and Liu, Yuanwei and Yang, Hongwen and Al-Dhahir, Naofal},
  journal={IEEE Commun. Mag.},
  volume={61},
  number={5},
  pages={26--32},
  year={2023},
  month=may,
  publisher={IEEE}
}

@article{pizzo2022spatial,
  title={Spatial characterization of electromagnetic random channels},
  author={Pizzo, Andrea and Sanguinetti, Luca and Marzetta, Thomas L},
  journal={IEEE Open J. Commun. Soc.},
  volume={3},
  pages={847--866},
  year={2022},
  publisher={IEEE}
}

@ARTICLE{boqun_jstsp,
	title={Modeling and Analysis of Near-Field {ISAC}}, 
	author={Zhao, Boqun and Ouyang, Chongjun and Liu, Yuanwei and Zhang, Xingqi and Poor, H. Vincent},
	journal={IEEE J. Sel. Top. Signal Process.},
	year={2024},
	volume={18},
	number={4},
	pages={678-693}
}

@article{liu2024capa,
  title={{CAPA}: Continuous-aperture arrays for revolutionizing {6G} wireless communications},
  author={Liu, Yuanwei and Ouyang, Chongjun and Wang, Zhaolin and Xu, Jiaqi and Mu, Xidong and Ding, Zhiguo},
  journal={IEEE Wireless Commun.},
    year={2025},
    month={Aug.},
  volume={32},
  number={4},
  pages={38-45},
}

@article{opt2,
				title={Pattern-division multiplexing for multi-user continuous-aperture {MIMO}},
				author={Zhang, Zijian and Dai, Linglong},
				journal={IEEE J. Sel. Areas Commun.},
				volume={41},
				number={8},
				pages={2350--2366},
				year={2023},
				month={Aug.},
				publisher={IEEE}
			}

@ARTICLE{opt3,
			author={Qian, Mengyu and You, Li and Xia, Xiang-Gen and Gao, Xiqi},
			journal={{IEEE} Trans. Wireless Commun.}, 
			title={On the Spectral Efficiency of Multi-User Holographic {MIMO} Uplink Transmission}, 
			year={2024},
			month={Oct.},
			volume={23},
			number={10},
			pages={15421-15434},
			doi={10.1109/TWC.2024.3429495}}

@article{opt4,
  title={Optimal Beamforming for Multi-User Continuous Aperture Array ({CAPA}) Systems},
  author={Wang, Zhaolin and Ouyang, Chongjun and Liu, Yuanwei},
  journal={IEEE Trans. Commun.},
  year={2025},
  month={Oct.},
  volume={73},
  number={10},
  pages={9207-9221}
}

@article{opt5,
					title={Beamforming Optimization for Continuous Aperture Array ({CAPA})-based Communications},
					author={Wang, Zhaolin and Ouyang, Chongjun and Liu, Yuanwei},
					journal={IEEE Trans. Wireless Commun.},
					  year={2025},
					volume={24},
					number={6},
					pages={5099-5113}
				}

@ARTICLE{per3,
				author={Wan, Zhongzhichao and Zhu, Jieao and Dai, Linglong},
				journal={IEEE Commun. Lett.}, 
				title={Can Continuous Aperture {MIMO} Obtain More Mutual Information Than Discrete {MIMO}?}, 
				year={2023},
				month={Dec.},
				volume={27},
				number={12},
				pages={3185-3189},
				doi={10.1109/LCOMM.2023.3329134}}

@article{ouyang_diversity,
		title={Diversity and Multiplexing for Continuous Aperture Array ({CAPA})-Based Communications},
		author={Ouyang, Chongjun and Wang, Zhaolin and Zhang, Xingqi and Liu, Yuanwei},
		journal={IEEE Trans. Wireless Commun.},
		year={2025},
		volume={24},
		number={10},
		pages={8662-8680}
	}

@ARTICLE{LiuAn,
  author={Liu, An and others},
  journal={IEEE Commun. Surveys Tuts.}, 
  title={A Survey on Fundamental Limits of Integrated Sensing and Communication}, 
  year={2022},
  month={Feb.},
  volume={24},
  number={2},
  pages={994-1034},
  doi={10.1109/COMST.2022.3149272}}

@ARTICLE{holoISAC_2,
  author={Zhang, Haobo and others},
  journal={{IEEE} J. Sel. Areas Commun.}, 
  title={Holographic Integrated Sensing and Communication}, 
  year={2022},
month={Jul.},
  volume={40},
  number={7},
  pages={2114-2130},
  doi={10.1109/JSAC.2022.3155548}}

@article{liu2022integrated,
  title={Integrated sensing and communications: Toward dual-functional wireless networks for {6G} and beyond},
  author={Liu, Fan and Cui, Yuanhao and Masouros, Christos and Xu, Jie and Han, Tony Xiao and Eldar, Yonina C and Buzzi, Stefano},
  journal={IEEE J. Sel. Areas Commun.},
  volume={40},
  number={6},
  pages={1728--1767},
  year={2022},
  month=jun,
  publisher={IEEE}
}

@article{pizzo2022nyquist,
  title={Nyquist sampling and degrees of freedom of electromagnetic fields},
  author={Pizzo, Andrea and de Jesus Torres, Andrea and Sanguinetti, Luca and Marzetta, Thomas L},
  journal={IEEE Trans. Signal Process.},
  volume={70},
  pages={3935--3947},
  year={2022},
  publisher={IEEE}
}

@article{chen2023cramer,
  title={{Cram{\'e}r-Rao} bounds of near-field positioning based on electromagnetic propagation model},
  author={Chen, Ang and Chen, Li and Chen, Yunfei and You, Changsheng and Wei, Guo and Yu, F Richard},
  journal={IEEE Trans. Veh. Technol.},
  volume={72},
  number={11},
  pages={13808--13825},
  year={2023},
  month=nov,
  publisher={IEEE}
}

@article{chen2024near,
  title={Near-field positioning and attitude sensing based on electromagnetic propagation modeling},
  author={Chen, Ang and Chen, Li and Chen, Yunfei and Zhao, Nan and You, Changsheng},
  journal={IEEE J. Sel. Areas Commun.},
  volume={42},
  number={9},
  pages={2179--2195},
  year={2024},
  month=sep,
  publisher={IEEE}
}

@article{jiang2024cram,
  title={{Cram{\'e}r-Rao} bounds Optimization for Near-Field Sensing with Continuous-Aperture Arrays},
  author={Jiang, Hao and Wang, Zhaolin and Liu, Yuanwei and Nallanathan, Arumugam},
  journal={IEEE Trans. Wireless Commun.},
  year={2026},
  volume={25},
  number={},
  pages={7032-7047}
}

@article{landau1980eigenvalue,
	title={Eigenvalue distribution of time and frequency limiting},
	author={Landau, Henry J and Widom, Harold},
	journal={J. Math. Anal. and Appl.},
	volume={77},
	number={2},
	pages={469--481},
	year={1980},
	publisher={Academic Press}
}

@ARTICLE{landau_tit,
	author={Franceschetti, Massimo},
	journal={IEEE Trans. Inf. Theory}, 
	title={On Landau’s Eigenvalue Theorem and Information Cut-Sets}, 
	year={2015},
	month={Sep.},
	volume={61},
	number={9},
	pages={5042-5051},
	doi={10.1109/TIT.2015.2456874}}

@article{sum_gamma,
		title={The distribution of the sum of independent gamma random variables},
		author={Moschopoulos, Peter G},
		journal={Ann. Inst. Statist. Math.},
		volume={37},
		number={1},
		pages={541--544},
		year={1985}
	}

@article{zhao2025downlink,
    	title={Downlink and uplink {ISAC} in continuous-aperture array ({CAPA}) systems},
    	author={Zhao, Boqun and Ouyang, Chongjun and Zhang, Xingqi and Shin, Hyundong and Liu, Yuanwei},
    	journal={IEEE Trans. Wireless Commun.},
    	year={2026},
    	volume={25},
    	number={},
    	pages={3592-3609}
    }

@ARTICLE{bcrb,
    	author={Xie, Lei and Liu, Fan and Luo, Jiajin and Song, Shenghui},
    	journal={IEEE Trans. Commun.}, 
    	title={Sensing Mutual Information with Random Signals in {Gaussian} Channels}, 
    	 year={2025},
    	 month={Oct.},
    	volume={73},
    	number={10},
    	pages={9437-9452}}

@ARTICLE{dardari,
    	author={Dardari, Davide},
    	journal={IEEE J. Sel. Areas Commun.}, 
    	title={Communicating With Large Intelligent Surfaces: Fundamental Limits and Models}, 
    	year={2020},
    	month={Nov.},
    	volume={38},
    	number={11},
    	pages={2526-2537},
    	doi={10.1109/JSAC.2020.3007036}}

@ARTICLE{zhangyue,
    	author={Zhang, Yue and others},
    	journal={IEEE Trans. Wireless Commun.}, 
    	title={Exploiting Continuous-Aperture Arrays in Integrated Sensing and Communication Systems}, 
    	year={2025},
    	month={Nov.},
    	volume={24},
    	number={11},
    	pages={9539-9555},
    	doi={10.1109/TWC.2025.3573872}}

@ARTICLE{qian,
    	author={Qian, Mengyu and Mu, Xidong and You, Li and Matthaiou, Michail},
    	journal={IEEE Trans. Commun.}, 
    	title={Continuous Aperture Array ({CAPA})-Based Multi-Group Multicast Communications}, 
    	year={2026},
    	month={Jan.},
    	volume={74},
    	number={},
    	pages={3787-3801},
    	doi={10.1109/TCOMM.2026.3655762}}

@article{chen2025implicit,
    	title={Implicit Neural Representation of Beamforming for Continuous Aperture Array ({CAPA}) System},
    	author={Chen, Shiyong and Guo, Jia and Han, Shengqian},
    	journal={arXiv preprint arXiv:2507.03609},
    	year={2025}
    }

@ARTICLE{gan2024coverage,
    	author={Gan, Xu and Huang, Chongwen and Yang, Zhaohui and Chen, Xiaoming and He, Jiguang and Zhang, Zhaoyang and Yuen, Chau and Guan, Yong Liang and Debbah, M{\'e}rouane},
    	journal={{IEEE} J. Sel. Areas Commun.},
    	title={Coverage and Rate Analysis for Integrated Sensing and Communication Networks},
    	year={2024},
    	month={Sep.},
    	volume={42},
    	number={9},
    	pages={2213-2227},
    	doi={10.1109/JSAC.2024.3413989}}

@ARTICLE{wei2025isac,
    	author={Wei, Zhiqing and Jia, Jinzhu and Niu, Yangyang and Wang, Lin and Wu, Huici and Yang, Heng and Feng, Zhiyong},
    	journal={{IEEE} Internet Things J.},
    	title={Integrated Sensing and Communication Channel Modeling: A Survey},
    	year={2025},
    	volume={12},
    	number={12},
    	pages={18850-18864},
    	doi={10.1109/JIOT.2024.3449377}}

@ARTICLE{bazzi2023outage,
    	author={Bazzi, Ahmad and Chafii, Marwa},
    	journal={{IEEE} Trans. Wireless Commun.},
    	title={On Outage-Based Beamforming Design for Dual-Functional Radar-Communication {6G} Systems},
    	year={2023},
    	month={Aug.},
    	volume={22},
    	number={8},
    	pages={5598-5612},
    	doi={10.1109/TWC.2023.3235617}}

@ARTICLE{meng2025cooperative,
    	author={Meng, Kaitao and Masouros, Christos and Petropulu, Athina P. and Hanzo, Lajos},
    	journal={{IEEE} Trans. Wireless Commun.},
    	title={Cooperative {ISAC} Networks: Performance Analysis, Scaling Laws, and Optimization},
    	year={2025},
    	month={Feb.},
    	volume={24},
    	number={2},
    	pages={877-892},
    	doi={10.1109/TWC.2024.3491356}}

@ARTICLE{liu2022cramer,
    	author={Liu, Fan and Liu, Ya-Feng and Li, Ang and Masouros, Christos and Eldar, Yonina C.},
    	journal={{IEEE} Trans. Signal Process.},
    	title={{Cram{\'e}r-Rao} Bound Optimization for Joint Radar-Communication Beamforming},
    	year={2022},
    	volume={70},
    	pages={240-253}}

@ARTICLE{dong2023joint,
    	author={Dong, Yong and Liu, Fan and Xiong, Yifeng},
    	journal={{IEEE} Commun. Lett.},
    	title={Joint Receiver Design for Integrated Sensing and Communications},
    	year={2023},
    	month={Jul.},
    	volume={27},
    	number={7},
    	pages={1854-1858}}

@ARTICLE{tang2019spectrally,
    	author={Tang, Bo and Li, Jian},
    	journal={{IEEE} Trans. Signal Process.},
    	title={Spectrally Constrained {MIMO} Radar Waveform Design Based on Mutual Information},
    	year={2019},
    	month={Feb.},
    	volume={67},
    	number={3},
    	pages={821-834}}

@ARTICLE{jiang2024electromagnetic,
    	author={Jiang, Yu and Gao, Feifei and Jin, Shi},
    	journal={{IEEE} Trans. Wireless Commun.},
    	title={Electromagnetic Property Sensing: A New Paradigm of Integrated Sensing and Communication},
    	year={2024},
    	month={Oct.},
    	volume={23},
    	number={10},
    	pages={13471-13483}}

@ARTICLE{wheeler1965,
    	author={Wheeler, Harold},
    	journal=IEEE_ANTEN,
    	title={Simple relations derived from a phased-array antenna made of an infinite current sheet},
    	year={1965},
    	month={Jul.},
    	volume={13},
    	number={4},
    	pages={506-514}}

@ARTICLE{huang2020holographic,
    	author={Huang, Chongwen and Hu, Sha and Alexandropoulos, George C. and Zappone, Alessio and Yuen, Chau and Zhang, Rui and Di Renzo, Marco and Debbah, M{\'e}rouane},
    	journal=IEEE_WM_COM,
    	title={Holographic {MIMO} surfaces for {6G} wireless networks: Opportunities, challenges, and trends},
    	year={2020},
    	month={Oct.},
    	volume={27},
    	number={5},
    	pages={118-125}}

@ARTICLE{shlezinger2021dynamic,
    	author={Shlezinger, Nir and Alexandropoulos, George C. and Imani, Mohammadreza F. and Eldar, Yonina C. and Smith, David R.},
    	journal=IEEE_WM_COM,
    	title={Dynamic metasurface antennas for {6G} extreme massive {MIMO} communications},
    	year={2021},
    	month={Apr.},
    	volume={28},
    	number={2},
    	pages={106-113}}

@ARTICLE{araghi2021holographic,
    	author={Araghi, Ali and Khalily, Mohsen and Xiao, Pei and Tafazolli, Rahim},
    	journal={{IEEE} Microw. Mag.},
    	title={Holographic-based leaky-wave structures: Transformation of guided waves to leaky waves},
    	year={2021},
    	month={Jun.},
    	volume={22},
    	number={6},
    	pages={49-63}}

@ARTICLE{prather2017optically,
    	author={Prather, Dennis W. and others},
    	journal=IEEE_ANTEN,
    	title={Optically upconverted, spatially coherent phased-array-antenna feed networks for beam-space {MIMO} in {5G} cellular communications},
    	year={2017},
    	month={Dec.},
    	volume={65},
    	number={12},
    	pages={6432-6443}}

@INPROCEEDINGS{sazegar2022full,
    	author={Sazegar, Mohammadreza and others},
    	booktitle={Proc. {IEEE} Int. Symp. Antennas Propag. USNC-URSI Radio Sci. Meeting (AP-S/URSI)},
    	title={Full duplex {SATCOM} {ESA} with switchable polarization and wide tunable bandwidth using a tripleband metasurface aperture},
    	year={2022},
    	pages={639-640}}

@MISC{pivotal2019holographic,
    	author={Staff, Pivotal},
    	title={Holographic beam forming and phased arrays},
    	howpublished={Pivotal Commware, Inc., Kirkland, WA, USA},
    	year={2019}}

@ARTICLE{smith2017analysis,
    	author={Smith, David R. and others},
    	journal={Phys. Rev. Appl.},
    	title={Analysis of a waveguide-fed metasurface antenna},
    	year={2017},
    	month={Nov.},
    	volume={8},
    	number={5},
    	pages={054048}}

@ARTICLE{yuan2024interdigital,
    	author={Yuan, Y. and Jiang, C. and Zhu, J.},
    	journal={Appl. Opt.},
    	title={Interdigital waveguide grating antenna array for an optical phased array},
    	year={2024},
    	month={Oct.},
    	volume={63},
    	number={28},
    	pages={7370-7377}}

@MISC{heath2025trihybrid,
    	author={Heath, R. W. and Carlson, J. and Deshpande, N. V. and Castellanos, M. R. and Akrout, M. and Chae, C.-B.},
    	title={The tri-hybrid {MIMO} architecture},
    	howpublished={arXiv:2505.21971},
    	year={2025}}

@ARTICLE{shao2025sixd,
    	author={Shao, X. and Jiang, Q. and Zhang, R.},
    	journal={{IEEE} Trans. Wireless Commun.},
    	title={{6D} movable antenna based on user distribution: Modeling and optimization},
    	year={2025},
    	month={Jan.},
    	volume={24},
    	number={1},
    	pages={355-370}}

@ARTICLE{li2026reconfig,
    	author={Li, Y. and Qi, C. and Mao, S. and Dobre, O. A.},
    	journal={{IEEE} Trans. Wireless Commun.},
    	title={Tri-hybrid beamforming for radiation-center reconfigurable antenna array: Spectral efficiency and energy efficiency},
    	year={2026},
    	volume={25},
    	pages={12263-12278}}
\end{document}